\def\isarxiv{1} 

\ifdefined\isarxiv
\documentclass[11pt]{article}

\usepackage[numbers]{natbib}

\else

\documentclass{article}
\usepackage{hyperref}
\usepackage{icml2023}

\fi

\usepackage{amsmath}
\usepackage{amsthm}
\usepackage{amssymb}
\usepackage{algorithm}
\usepackage{subfig}
\usepackage{algpseudocode}
\usepackage{graphicx}
\usepackage{grffile}
\usepackage{wrapfig,epsfig}
\usepackage{url}
\usepackage{xcolor}
\usepackage{epstopdf}
\usepackage{bbm}

\usepackage{dsfont} 

\allowdisplaybreaks

\ifdefined\isarxiv

\usepackage{tikz}
\usepackage{hyperref}  
\hypersetup{colorlinks=true,citecolor=blue,linkcolor=blue} 
\usetikzlibrary{arrows}
\usepackage[margin=1in]{geometry}

\else
\usepackage{microtype}

\fi

\ifdefined\isarxiv 

\else 

\fi

\newtheorem{theorem}{Theorem}[section]
\newtheorem{lemma}[theorem]{Lemma}
\newtheorem{definition}[theorem]{Definition}

\newtheorem{corollary}[theorem]{Corollary}

\newtheorem{fact}[theorem]{Fact}

\newtheorem{claim}[theorem]{Claim}

\newtheorem{problem}[theorem]{Problem}

\newcommand{\wt}{\widetilde}

\newcommand{\R}{\mathbb{R}}

\renewcommand{\varepsilon}{\epsilon}

\DeclareMathOperator*{\E}{{\mathbb{E}}}

\DeclareMathOperator{\poly}{poly}

\DeclareMathOperator{\tr}{tr}

\newcommand{\Zhao}[1]{{\color{red}[Zhao: #1]}}
\makeatletter
\newcommand*{\RN}[1]{\expandafter\@slowromancap\romannumeral #1@}
\makeatother
\newcommand{\Danyang}[1]{{\color{purple}[Danyang: #1]}}
\newcommand{\lianke}[1]{{\color{blue}[Lianke: #1]}}

\newcommand{\Ruizhe}[1]{{\color{orange}[Ruizhe: #1]}}

\usepackage{lineno}

\ifdefined\isarxiv

\else 
\fi

\begin{document}

\ifdefined\isarxiv

\date{}

\title{A General Algorithm for Solving Rank-one Matrix Sensing}
\author{
Lianke Qin\thanks{\texttt{lianke@ucsb.edu}. UCSB.}
\and 
Zhao Song\thanks{\texttt{zsong@adobe.com}. Adobe Research.}
\and 
Ruizhe Zhang\thanks{\texttt{ruizhe@utexas.edu}. The University of Texas at Austin.}
}

\else
\twocolumn[

\icmltitle{A General Algorithm for Solving Rank-one Matrix Sensing}


\icmlsetsymbol{equal}{*}

\begin{icmlauthorlist}
\icmlauthor{Aeiau Zzzz}{equal,to}
\icmlauthor{Bauiu C.~Yyyy}{equal,to,goo}
\icmlauthor{Cieua Vvvvv}{goo}
\icmlauthor{Iaesut Saoeu}{ed}
\icmlauthor{Fiuea Rrrr}{to}
\icmlauthor{Tateu H.~Yasehe}{ed,to,goo}
\icmlauthor{Aaoeu Iasoh}{goo}
\icmlauthor{Buiui Eueu}{ed}
\icmlauthor{Aeuia Zzzz}{ed}
\icmlauthor{Bieea C.~Yyyy}{to,goo}
\icmlauthor{Teoau Xxxx}{ed}\label{eq:335_2}
\icmlauthor{Eee Pppp}{ed}
\end{icmlauthorlist}

\icmlaffiliation{to}{Department of Computation, University of Torontoland, Torontoland, Canada}
\icmlaffiliation{goo}{Googol ShallowMind, New London, Michigan, USA}
\icmlaffiliation{ed}{School of Computation, University of Edenborrow, Edenborrow, United Kingdom}

\icmlcorrespondingauthor{Cieua Vvvvv}{c.vvvvv@googol.com}
\icmlcorrespondingauthor{Eee Pppp}{ep@eden.co.uk}

\icmlkeywords{Machine Learning, ICML}

\vskip 0.3in
]

\printAffiliationsAndNotice{\icmlEqualContribution} 
\fi

\ifdefined\isarxiv
\begin{titlepage}
  \maketitle
  \begin{abstract}
Matrix sensing has many real-world applications in science and engineering, such as system control, distance embedding, and computer vision. The goal of matrix sensing is to recover a  matrix $A_\star \in \R^{n \times n}$, based on a sequence of measurements $(u_i,b_i) \in \R^{n} \times \R$ such that $u_i^\top A_\star u_i = b_i$. Previous work \cite{zjd15} focused on the scenario where matrix $A_{\star}$ has a small rank, e.g. rank-$k$. Their analysis heavily relies on the RIP assumption, making it unclear how to generalize to high-rank matrices. In this paper, we relax that rank-$k$ assumption and solve a much more general matrix sensing problem. 
Given an accuracy parameter $\delta \in (0,1)$, we can compute $A \in \R^{n \times n}$ in $\widetilde{O}(m^{3/2} n^2 \delta^{-1} )$, such that $ |u_i^\top A u_i - b_i| \leq \delta$ for all $i \in [m]$. We design an efficient algorithm with provable convergence guarantees using stochastic gradient descent for this problem.

  \end{abstract}
  \thispagestyle{empty}
\end{titlepage}

{\hypersetup{linkcolor=black}
\tableofcontents
}
\newpage

\else
\begin{abstract}
Matrix sensing has many real-world applications in science and engineering, such as system control, distance embedding, and computer vision. The goal of matrix sensing is to recover a  matrix $A_\star \in \R^{n \times n}$, based on a sequence of measurements $(u_i,b_i) \in \R^{n} \times \R$ such that $u_i^\top A_\star u_i = b_i$. Previous work \cite{zjd15} focused on the scenario where matrix $A_{\star}$ has a small rank, e.g. rank-$k$. Their analysis heavily relies on the RIP assumption, making it unclear how to generalize to high-rank matrices. In this paper, we relax that rank-$k$ assumption and solve a much more general matrix sensing problem. 
Given an accuracy parameter $\delta \in (0,1)$, we can compute $A \in \R^{n \times n}$ in $\widetilde{O}(m^{3/2} n^2 \delta^{-1} )$, such that $ |u_i^\top A u_i - b_i| \leq \delta$ for all $i \in [m]$. We design an efficient algorithm with provable convergence guarantees using stochastic gradient descent for this problem.
 
\end{abstract}

\fi

\section{Introduction}

Matrix sensing is a generalization of the famous compressed sensing problem. Informally, the goal of matrix sensing is to reconstruct a matrix $A\in \R^{n\times n}$ using a small number of quadratic measurements (i.e., $u^\top A u$). It has many real-world applications, including image processing~\citep{clmw11, wsb11}, quantum computing~\citep{a07, fgle12,kkd15}, systems~\citep{lv10} and sensor localization~\citep{jm13} problems. 
For this problem, there are two important \emph{theoretical} questions:
\begin{itemize}
    \item {\bf Q1. Compression: } how to design the sensing vectors $u\in \R^n$ so that the matrix can be recovered with a small number of measurements?
    \item {\bf Q2. Reconstruction: } how fast can we recover the matrix given the measurements?
\end{itemize}

\cite{zjd15} initializes the study of the rank-one matrix sensing problem, where the ground-truth matrix $A_{\star}$ has only rank-$k$, and the measurements are of the form $u_i^\top A_{\star} u_i$. They want to know the smallest number of measurements $m$ to recover the matrix $A_{\star}$. In our setting, we assume $m$ is a fixed input parameter and we're not allowed to choose. We show that for any $m$ and $n$, how to design a faster algorithm for solving an optimization problem which is finding $A \approx A_{\star}$. Thus, in some sense, previous work \cite{zjd15,dls23} mainly focuses on problem {\bf Q1} with a low-rank assumption on $A_\star$. Our work is focusing on {\bf Q2} without the low-rank assumption.

We observe that in many applications, the ground-truth matrix $A_{\star}$ does not need to be recovered \emph{exactly} (i.e., $\|A-A_\star\|\leq n^{-c}$). For example, for distance embedding, we would like to learn an embedding matrix between all the data points in a high-dimensional space. The embedding matrix is then used for calculating data points' pairwise distances for a higher-level machine learning algorithm, such as $k$-nearest neighbor clustering. As long as we can recover a good approximation of the embedding matrix, the clustering algorithm can deliver the desired results. As we relax the accuracy constraints of the matrix sensing, we have the opportunity to speed up the matrix sensing time.

We formulate our problem in the following way:

\begin{problem}[Approximate matrix sensing]\label{prob:app_mat_sensing}
Given a ground-truth positive definite matrix $A_\star \in \R^{n \times n}$ and $m$ samples $(u_i,b_i) \in \R^{n} \times \R$ such that $u_i^\top A_\star u_i = b_i$. Let $R=\max_{i\in [m]} |b_i|$. For any accuracy parameter $\delta \in (0,1)$, find a matrix $A\in \R^{n\times n}$ such that
\begin{align}\label{eq:measure_guarantee}
    (u_i^\top A u_i - u_i A_{\star} u_i )^2 \leq \delta, ~~~\forall i \in [m]
\end{align}
or 
\begin{align}\label{eq:spectral_guarantee}
 (1-\delta ) A_\star \preceq A \preceq (1+\delta) A_\star.
\end{align}
\end{problem}
We make a few remarks about Problem~\ref{prob:app_mat_sensing}. First, our formulation doesn't require the matrix $A_{\star}$ to be low-rank as literature \cite{zjd15}. Second, we need the measurement vectors $u_i$ to be ``approximately orthogonal'' (i.e., $|u_i^\top u_j|$ are small), while \cite{zjd15} make much stronger assumptions for exact reconstruction. Third, the \emph{measure approximation} guarantee (Eq.~\eqref{eq:measure_guarantee}) does not imply the \emph{spectral approximation} guarantee (Eq.~\eqref{eq:spectral_guarantee}). We mainly focus on achieving the first guarantee and discuss the second one in the appendix.

This problem is interesting for two reasons. First, speeding up matrix sensing is salient for a wide range of applications, where exact matrix recovery is not required. Second, we would like to understand the fundamental tradeoff between the accuracy constraint $\epsilon$ and the running time. This tradeoff can give us insights on the fundamental computation complexity for matrix sensing.

This paper makes the following contributions:
\begin{itemize}
    \item We design a potential function to measure the distance between the approximate solution and the ground-truth matrix. 
    \item Based on the potential function, we show that gradient descent can efficiently find an approximate solution of the matrix sensing problem. We also prove the convergence rate of our algorithm. 
    \item Furthermore, we show that the cost-per-iteration can be improved by using stochastic gradient descent with a provable convergence guarantee, which is proved by generalizing the potential function to a randomized potential function.
\end{itemize}

Technically, our potential function applies a $\cosh$ function to each ``training loss'' (i.e., $u_i^\top Au_i - b_i$), which is inspired by the potential function for linear programming \citep{cls19}. We prove that the potential is decreasing for each iteration of gradient descent, and a small potential implies a good approximation. In this way, we can upper bound the number of iterations needed for the gradient descent algorithm.

To reduce the cost-per-iteration, we follow the idea of stochastic gradient descent and evaluate the gradient of potential function on a subset of measurements. However, we still need to know the full gradient's norm for normalization, which is a function of the training losses. It is too slow to naively compute each training loss. Instead, we use the idea of maintenance \citep{cls19,lsz19,b20,blss20,jswz21,hjstz21,syz21,szz21,hswz22,qszz23} and show that the training loss at the $(t+1)$-th iteration (i.e., $u_i^\top A_{t+1}u_i-b_i$) can be very efficiently obtained from those at the $t$-th iteration (i.e., $u_i^\top A_{t}u_i-b_i$). Therefore, we first preprocess the initial full gradient's norm, and in the following iterations, we can update this quantity based on the previous iteration's result.

We state our main result as follows:
\begin{theorem}[Informal of Theorem~\ref{thm:sgd_orthogonal}]
Given $m$ measurements of matrix sensing problems, there is an algorithm that outputs a $n \times n$ matrix $A$ in $\wt{O}(m^{3/2} n^2 R\delta^{-1} )$ time such that $|u_i^\top A u_i - b_i|\leq \delta$, $\forall i \in [m]$.
\end{theorem} 
 
\section{Related Work}

\paragraph{Linear Progamming}
Linear programming is one of foundations of the algorithm design and convex optimization. many problems can be modeled as linear programs to take advantage of fast algorithms.
There are many works in accelerating linear programming runtime complexity~\citep{ls14, ls15, cls19, lsz19, b20, blss20, sy21, dly21, jswz21,gs22}.

\paragraph{Semi-definite Programming}

Semidefinite programming optimizes a linear objective function over the intersection of the
positive semidefinite cone with an affine space. Semidefinite programming is a fundamental class of
optimization problems and many problems in machine learning, and theoretical computer science
can be modeled or approximated as semidefinite programming problems. There are many studies to speedup the running
time of Semidefinite programming~\citep{nn94, hrvw96, lsw15, jlsw20, jklps20, hjstz21, gs22}.

\paragraph{Matrix Sensing}

Matrix sensing~\citep{lb09, rfp10, jmd10, zjd15, dls23} is a generalization of the popular compressive sensing problem for the sparse vectors and has applications in several
domains such as control, vision etc.
a set of universal Pauli measurements,
used in quantum state tomography, have been shown to satisfy the RIP condition~\citep{l11}.
These measurement operators are Kronecker products of $2 \times 2$ matrices, thus, they have appealing computation and memory efficiency. Rank-one measurement using nuclear norm minimization is also used in other work~\citep{cz15, krt17}.
There is also previous work working on low-rank matrix sensing to reconstruct a matrix exactly using a small number of linear measurements.
ProcrustesFlow~\cite{tbssr16} designs an algorithm to recover a low-rank matrix from linear measurements.
There are other low-rank matrix recovering algorithms based on non-convex optimizations~\cite{ wzg17, lmcc19}.

\section{Preliminary}\label{sec:preli}
\paragraph{Notations.}For a positive integer, we use $[n]$ to denote set $\{ 1,2,\cdots,n\}$. 
We use $\cosh(x) =\frac{1}{2}( e^x + e^{-x})$ and $\sinh(x) = \frac{1}{2}(e^x - e^{-x} )$.
For a square matrix, we use $\tr[A]$ to denote the trace of $A$.
An $n \times n$ symmetric real matrix $A$ is said to be positive-definite if $x^{\top} A x > 0$ for all non-zero $x \in \R^n$.
An $n \times n$ symmetric real matrix $A$ is said to be positive-semidefinite if $x^{\top} A x \geq 0$ for all non-zero $x \in \R^n$. For any function $f$, we use $\wt{O}(f) = f \cdot \poly(\log f)$.

\subsection{Matrix hyperbolic functions}
\begin{definition}[Matrix function]
Let $f:\R \rightarrow \R$ be a real function and $A\in \R^{n\times n}$ be a real symmetric function with eigendecomposition 
\begin{align*} 
A=Q\Lambda Q^{-1}
\end{align*}
where $\Lambda\in \R^{n\times n}$ is a diagonal matrix. Then, we have
\begin{align*}
    f(A):=Qf(\Lambda) Q^{-1},
\end{align*}
where $f(\Lambda)\in \R^{n\times n}$ is the matrix obtained by applying $f$ to each diagonal entry of $\Lambda$.
\end{definition}

We have the following lemma to bound $\cosh(A)$ and delay the proof to Appendix~\ref{sec:cosh_bound_proof}.
\begin{lemma}\label{lem:cosh_bound}
Let $A$ be a real symmetric matrix, then we have
\begin{align*}
    \|\cosh(A)\| = \cosh(\|A\|) \leq \tr[\cosh(A)].
\end{align*}
We also have 
\begin{align*}
\|A\| \leq 1+\log(\tr[\cosh(A)]).
\end{align*}
\end{lemma}

\subsection{Properties of \texorpdfstring{$\sinh$}{} 
and \texorpdfstring{$\cosh$}{}
}

We have the following lemma for properties of $\sinh$ and $\cosh$. 
\begin{lemma}[Scalar version]\label{lem:property_sinh_cosh_scalar}
Given a list of numbers $x_1, \cdots x_n$, we have
\begin{itemize}
    \item $( \sum_{i=1}^n \cosh^2(x_i) )^{1/2} \leq \sqrt{n} + ( \sum_{i=1}^n \sinh^2(x_i) )^{1/2}$,
    \item $(\sum_{i=1}^n \sinh^2(x_i) )^{1/2} \geq \frac{1}{\sqrt{n}} (\sum_{i=1}^n \cosh(x_i) - n)$.
\end{itemize}
\end{lemma}
\begin{proof}
For the first equation, we can bound $( \sum_{i=1}^n \cosh^2(x_i) )^{1/2}$ by:
\begin{align*}
     ( \sum_{i=1}^n \cosh^2(x_i) )^{1/2} 
    = &~ (n + \sum_{i=1}^n \sinh^2(x_i))^{1/2} \\
    \leq &~\sqrt{n} + ( \sum_{i=1}^n \sinh^2(x_i) )^{1/2}
\end{align*}
where the first step comes from fact~\ref{fact:cosh_sinh_1}, and the second step follows from $\sqrt{a + b} \leq \sqrt{a} + \sqrt{b}$.

For the second equation, we can bound $(\sum_{i=1}^n \sinh^2(x_i) )^{1/2}$ by:
\begin{align*}
    (\sum_{i=1}^n \sinh^2(x_i) )^{1/2} 
    \geq &~ \frac{1}{\sqrt{n}}(\sum_{i=1}^{n} \sinh(x_i)) \\
    \geq &~ \frac{1}{\sqrt{n}}(\sum_{i=1}^{n} \cosh(x_i) -n)
\end{align*}
where the first step follows that $\sqrt{\frac{\sum_{i=1}^{n} x_i^2}{n}} \geq \frac{\sum_{i=1}^{n} x_i}{n}$,
and the second step follows from fact~\ref{fact:cosh_sinh_1} and $\sqrt{x^2 -1} \geq \sqrt{x} - 1$.
\end{proof}

We also have a lemma for the matrix version. 
\begin{lemma}[Matrix version]\label{lem:property_sinh_cosh_matrix}
For any real symmetric matrix $A$, we have
\begin{itemize}
    \item $ (\tr[\cosh^2(A)])^{1/2} \leq \sqrt{n} + \tr[ \sinh^2(A) ]^{1/2}$,
    \item $(\tr[ \sinh^2(A) ])^{1/2} \geq \frac{1}{\sqrt{n}} ( \tr[ \cosh(A) ] - n ) $.
\end{itemize}
\end{lemma}

\begin{proof}

{\bf Part 1.}
We have
\begin{align*}
    (\tr[\cosh^2(A)])^{1/2} = & ~ ( n+  \tr[\sinh^2(A)] )^{1/2} \\
    \leq & ~ \sqrt{n} + \tr[ \sinh^2(A) ]^{1/2}.
\end{align*}
where the first step follows from $ \cosh^2(A) - \sinh^2(A) = I$.

{\bf Part 2.}
Let $\sigma_i$ denote the singular value of $\cosh(A)$
\begin{align*}
    ( \tr[  \sinh^2(A) ] )^{1/2} 
    = & ~  ( \tr[ \cosh^2(A) ] - n )^{1/2} \\
    = & ~  ( \sum_{i=1}^n \sigma_i^2 - 1 )^{1/2} \\ 
    \geq & ~   \frac{1}{\sqrt{n}} \sum_{i=1}^n \sqrt{ \sigma_i^2  -1 } \\
    \geq & ~    \frac{1}{\sqrt{n}} (\sum_{i=1}^n \sigma_i - 1 ) \\
    = & ~ \frac{1}{\sqrt{n}} ( \tr[  \cosh(A) ] -  n ) 
\end{align*}
where the second step follows from $\| \cdot \|_2 \geq \frac{1}{\sqrt{n}} \| \cdot \|_1$, the third step follows from $\sigma_i \geq 1$.

\end{proof}

\section{Technique Overview}

We first analyze the convergence guarantee of our matrix sensing algorithm based on gradient descent and improve its time complexity with stochastic gradient descent under the assumption where $\{u_i\}_{i\in [m]}$ are orthogonal vectors. We then analyze the convergence guarantee of our matrix sensing algorithm under a more general assumption where $\{u_i\}_{i\in [m]}$ are non-orthogonal vectors and $|u_i^{\top}  u_j| \leq \rho$.

\paragraph{Gradient descent.} We begin from the case where $\{u_i\}_{i\in [m]}$ are orthogonal vectors in $\R^n$. We can the following entry-wise potential function:
\begin{align*}
   \Phi_{\lambda}(A) := \sum_{i=1}^m \cosh ( \lambda ( u_i^\top A u_i - b_i ) ) 
\end{align*}
and analyze its progress during the gradient descent according to the update formula defined in Eq.~\eqref{eq:A_t1_update} 
for each iteration. 
We split the gradient of the potential function into diagonal and off-diagonal terms. We can upper bound the diagonal term and prove that the off-diagonal term is zero.
Combining the two terms together, we can upper bound the progress of update per iteration in Lemma~\ref{lem:gradient_descent} by:
\begin{align*}
     \Phi_{\lambda} ( A_{t+1} ) \leq (1-0.9 \frac{ \lambda \epsilon }{\sqrt{m} }) \cdot \Phi_{\lambda} (A_t) +  \lambda \epsilon \sqrt{m}.
\end{align*}
By accumulating the progress of update for the entry-wise potential function over $T=\widetilde{\Omega}(\sqrt{m}R\delta^{-1})$ iterations, we have $\Phi(A_{T+1}) \leq O(m)$. This implies that our Algorithm~\ref{alg:GD} can output a matrix $A_{T} \in \R^{n \times n}$ satisfying guarantee in Eq.~\eqref{eq:gd_approximation_guarantee}, 
and the corresponding time complexity is $O(mn^2)$.

We then analyze the gradient descent under the assumption where $\{u_i\}_{i\in [m]}$ are non-orthogonal vectors in $\R^n$, $|u_i^{\top}  u_j| \leq \rho$ and $\rho \leq \frac{1}{10m}$. We can upper bound the diagonal entries and off-diagonal entries respectively and obtain the same progress of update per iteration in Lemma~\ref{lem:gradient_descent_rho}. Accumulating in $T=\widetilde{\Omega}(\sqrt{m}R\delta^{-1})$ iterations, we can prove the approximation guarantee of the output matrix of our matrix sensing algorithm.

\paragraph{Stochastic gradient descent.} To further improve the time cost per iteration of our approximate matrix sensing,  by uniformly sampling a subset ${\cal B}\subset [m]$ of size $B$, we compute the gradient of the stochastic potential function:
\begin{align*}
   \nabla \Phi_{\lambda}(A{ , {\cal B}}) := {  \frac{m}{|{\cal B}|}\sum_{i \in {\cal B}}} u_i u_i^\top \lambda \sinh( \lambda (u_i^\top A u_i - b_i) ),
\end{align*}
 and update the potential function based on update formula defined in Eq.~\eqref{eq:sgd_potential_func}. We upper bound the diagonal and off-diagonal terms respectively and obtain the expected progress on the potential function in Lemma~\ref{lem:stochastic_gradient_descent}.

    Over $T=\widetilde{\Omega}(m^{3/2}B^{-1}R\delta^{-1})$ iterations, we can upper bound $\Phi(A_{T+1}) \leq O(m)$ with high probability. With similar argument to gradient descent section, we can prove that the SGD matrix sensing algorithm can output a solution matrix satisfying the same approximation guarantees with high success probability in Lemma~\ref{lem:sgd_convergence}. The optimized time complexity is $O(Bn^2)$ where $B$ is the SGD batch size. 
    
For the more general assumption where $\{u_i\}_{i\in [m]}$ are non-orthogonal vectors in $\R^n$ and  $|u_i^{\top}  u_j|$ has an upper bound, 
We also provide the cost-per-iteration analysis for stochastic gradient descent by bounding the diagonal entries and off-diagonal entries of the gradient matrix respectively. Then we prove that the progress on the expected potential satisfies the same guarantee as the gradient descent in Lemma~\ref{lem:sgd_potential_general}. Therefore, our SGD matrix sensing algorithm can output an matrix satisfying the approximation guarantee after
\begin{align*} 
T=\widetilde{\Omega}(m^{3/2}B^{-1}R\delta^{-1})
\end{align*}
iterations under the general assumption.

\section{Gradient descent for entry-wise potential function}\label{sec:gd}

In this section, we show how to obtain an approximate solution of matrix sensing via gradient descent. For simplicity, we start from a case that $\{u_i\}_{i\in [m]}$ are orthogonal vectors in $\R^n$\footnote{We note that $A':=\sum_{i=1}^m b_iu_iu_i^\top$ is a solution satisfying $u_i^\top A'u_i = b_i$ for all $i\in [m]$. However, we pretend that we do not know this solution in this section.}, which already conveys the key idea of our algorithm and analysis and we generalize the solution to the non-orthogonal case (see Appendix~\ref{sec:gd_non_orthogonal}). 
We show that $\wt{\Omega}(\sqrt{m}/\delta)$ iterations of gradient descent can output a $\delta$-approximate solution, where each iteration takes $O(mn^2)$-time. Below is the main theorem of this section:

\begin{theorem}[Gradient descent for orthogonal measurements]\label{thm:gd_orthogonal}
Suppose $u_1,\dots,u_m\in \R^n$ are orthogonal unit vectors, and suppose $|b_i|\leq R$ for all $i\in [m]$. There exists an algorithm such that for any $\delta \in (0,1)$, performs $\wt{\Omega}(\sqrt{m}R\delta^{-1})$ iterations of gradient descent with $O(mn^2)$-time per iteration and outputs a matrix $A\in \R^{n\times n}$ satisfies:
\begin{align*}
    | u_i^\top A u_i - b_i| \leq \delta~~~\forall i\in [m].
\end{align*}
\end{theorem}

In Section~\ref{sec:gd_algorithm}, we introduce the algorithm and prove the time complexity. In Section~\ref{sec:gd_analysis_1} - \ref{sec:gd_analysis_2}, we analyze the convergence of our algorithm.

\subsection{Algorithm}\label{sec:gd_algorithm}

The key idea of the gradient descent matrix sensing algorithm (Algorithm~\ref{alg:GD}) is to follow the gradient of the entry-wise potential function defined as follows:
\begin{align}
    \Phi_{\lambda}(A) := \sum_{i=1}^m \cosh ( \lambda ( u_i^\top A u_i - b_i ) ).
\end{align}
Then, we have the following solution update formula:
\begin{align}\label{eq:A_t1_update}
    A_{t+1} \gets A_t - \epsilon \cdot \nabla \Phi_{\lambda}(A_t) / \| \nabla \Phi_{\lambda}(A_t) \|_F.
\end{align}

\begin{lemma}[Cost-per-iteration of gradient descent]\label{lem:gd_cost_per_iter}
Each iteration of Algorithm~\ref{alg:GD} takes
$
    O(mn^2)
$-time.
\end{lemma}
\begin{proof}
In each iteration, we first evaluate $u_i^\top A_tu_i$ for all $i\in [m]$, which takes $O(mn^2)$-time. Then, $\nabla \Phi_\lambda(A_t)$ can be computed by summing $m$ rank-1 matrices, which takes $O(mn^2)$-time. Finally, at Line~\ref{ln:gd_update}, the solution can be updated in $O(n^2)$-time.
Thus, the total running time for each iteration is $O(mn^2)$.
\end{proof}

\begin{algorithm}\caption{Matrix Sensing by Gradient Descent.}\label{alg:GD}
\begin{algorithmic}[1]
\Procedure{GradientDescent}{$\{u_i,b_i\}_{i\in [m]}$} \Comment{Theorem~\ref{thm:gd_orthogonal}}
    \State $\tau \gets \max_{i \in [m]} b_i $
    \State $A_1 \gets \tau \cdot I$
    \For{$t = 1 \to T$}
        \State $\nabla \Phi_{\lambda} (A_t) \gets \sum_{i=1}^m u_i u_i^\top \lambda \sinh( \lambda ( u_i^\top A_t u_i - b_i ) ) $ \Comment{Compute the gradient}
        \State $A_{t+1} \gets A_t - \epsilon \cdot \nabla \Phi_{\lambda}(A_t) / \| \nabla \Phi_{\lambda}(A_t) \|_F$\label{ln:gd_update}
    \EndFor
    \State \Return $A_{T+1}$
\EndProcedure
\end{algorithmic}
\end{algorithm}

\subsection{Analysis of One Iteration}\label{sec:gd_analysis_1}
Throughout this section, we suppose $A\in \R^{n\times n}$ is a symmetric matrix.

We can compute the gradient of $\Phi_{\lambda}(A)$ with respect to $A$ as follows:
\begin{align}\label{eq:gradient_phi_A}
    & ~ \nabla \Phi_{\lambda}(A) \notag \\
    = & ~ \sum_{i=1}^m u_i u_i^\top \lambda \sinh\left( \lambda (u_i^\top A u_i - b_i)\right) \in \R^{n\times n}. 
\end{align}

We can compute the Hessian of $\Phi_{\lambda}(A)$ with respect to $A$ as follows
\begin{align*}
    & ~ \nabla^2 \Phi_{\lambda}(A) \notag \\
    = & ~ \sum_{i=1}^m ( u_i u_i^\top ) \otimes ( u_i u_i^\top ) \lambda^2 \cosh( \lambda (u_i^\top A u_i - b_i) ). 
\end{align*}
The Hessian $\nabla^2 \Phi_{\lambda}(A)\in \R^{n^2\times n^2}$ and $\otimes$ is the Kronecker product.

\begin{lemma}[Progress on entry-wise potential]\label{lem:gradient_descent}
Assume that $u_i \perp u_j = 0 $ for any $i,j \in [m]$ and $\| u_i\|^2 = 1$. Let $c \in (0,1)$ denote a sufficiently small positive constant. Then, for any $\epsilon,\lambda>0$ such that $\epsilon\lambda \leq c$,
 
we have for any $t>0$,
\begin{align*}
    \Phi_{\lambda} ( A_{t+1} ) \leq (1-0.9 \frac{ \lambda \epsilon }{\sqrt{m} }) \cdot \Phi_{\lambda} (A_t) +  \lambda \epsilon \sqrt{m}
\end{align*}

\end{lemma}
\begin{proof}
We defer the proof to Appendix~\ref{sec:proof_gradient_descent}
\end{proof}

\subsection{Technical Claims}
We prove some technical claims in below.

\begin{claim}\label{cla:gd_Q1}
For $Q_1$ defined in Eq.~\eqref{eq:def_Q1}, we have
\begin{align*}
    Q_1 \leq \Big( \sqrt{m} + \frac{1}{\lambda} \| \nabla \Phi_{\lambda}(A_t) \|_F \Big) \cdot  \| \nabla \Phi_{\lambda} (A_t) \|_F^2.
\end{align*}
\end{claim}
\begin{proof}

For simplicity, we define $z_{t,i}$ to be
\begin{align*}
    z_{t,i} := \lambda ( u_i^\top A_t u_i - b_i ) .
\end{align*}
Recall that
\begin{align*}
    \nabla^2 \Phi_{\lambda}(A_t) = \lambda^2 \cdot \sum_{i=1}^m ( u_i u_i^\top ) \otimes ( u_i u_i^\top ) \cosh( z_{t,i} ) .
\end{align*}

For $Q_1$, we have
\begin{align}\label{eq:Q1}
   Q_1 = & ~  \tr[ \nabla^2 \Phi_{\lambda}(A_t) \sum_{i=1}^m  \sinh^2( z_{t,i} ) (u_i u_i^\top \otimes u_i u_i^\top) ) ] \notag\\
   = & ~ \lambda^2 \cdot \tr[ \sum_{i=1}^m  \cosh( z_{t,i} ) ( u_i u_i^\top ) \otimes ( u_i u_i^\top )    \cdot   \sum_{i=1}^m  \sinh^2( z_{t,i} ) (u_i u_i^\top ) \otimes ( u_i u_i^\top)   ] \notag\\
   = & ~ \lambda^2 \cdot \sum_{i=1}^m \tr[ \cosh( z_{t,i} )   \cdot \sinh^2( z_{t,i} )  ( u_i u_i^\top  u_i u_i^\top ) \otimes ( u_i u_i^\top  u_i u_i^\top ) ] \notag\\
   = & ~ \lambda^2 \cdot \sum_{i=1}^m  \cosh( z_{t,i} ) \sinh^2( z_{t,i} ) \notag \\
   \leq & ~ \lambda^2 \cdot (  \sum_{i=1}^m  \cosh^2( z_{t,i} ) )^{1/2}
   \cdot  ( \sum_{i=1}^m \sinh^4( z_{t,i} ) )^{1/2} \notag \\
   \leq & ~ \lambda^2 \cdot B_1 \cdot B_2,
\end{align}
where {  the first step comes from the definition of $Q_1$, the second step comes from the definition of $\nabla^2 \Phi_{\lambda}(A_t)$,}
the third step follows from $(A \otimes B) \cdot (C \otimes D) = (AC) \otimes (BD)$ and $u_i^\top u_j = 0$ , the fourth step comes from $\|u_i\| = 1$ and $\tr[ (u_i  u_i^\top) \otimes (u_i  u_i^\top) ] = 1$.

For the term $B_1$, we have
\begin{align*}
    B_1 = & ~ (  \sum_{i=1}^m \cosh^2( \lambda (u_i^\top A_t u_i - b_i ) ) )^{1/2} \\
    \leq & ~ \sqrt{m} + \frac{1}{\lambda} \| \nabla \Phi_{\lambda}(A_t) \|_F,
\end{align*}
where the second step follows Part 1 of Lemma~\ref{lem:property_sinh_cosh_scalar}.

For the term $B_2$, we have
\begin{align*}
    B_2 = & ~( \sum_{i=1}^m \sinh^4( \lambda( u_i^\top A_t u_i - b_i ) ) )^{1/2} \\
    \leq & ~ \frac{1}{\lambda^2} \| \nabla \Phi_{\lambda} (A_t) \|_F^2,
\end{align*}
where the second step follows from $\| x \|_4^2 \leq \| x \|_2^2$. This implies that
\begin{align*}
    Q_1 \leq & ~ \lambda^2 \cdot B_1 \cdot B_2 \\
    \leq & ~ \lambda^2 \cdot ( \sqrt{m} + \frac{1}{\lambda} \| \nabla \Phi_{\lambda}(A_t) \|_F ) \cdot \frac{1}{\lambda^2} \| \nabla \Phi_{\lambda} (A_t) \|_F^2 \\
    = & ~ ( \sqrt{m} + \frac{1}{\lambda} \| \nabla \Phi_{\lambda}(A_t) \|_F ) \cdot  \| \nabla \Phi_{\lambda} (A_t) \|_F^2 .
\end{align*}
\end{proof}

\begin{claim}\label{cla:gd_Q2}
For $Q_2$ defined in Eq.~\eqref{eq:def_Q2}, we have $Q_2 = 0$.
\end{claim}

\begin{proof}
Because in $Q_2$ we have :
\begin{align}\label{eq:u_ell_i_j_product}
    & ~\sum_{\ell=1}^{m}(u_{\ell} u_{\ell}^{\top} \otimes u_{\ell} u_{\ell}^{\top}) \sum_{i \neq j}(u_i u_i^{\top} \otimes u_j u_j^{\top}) \notag \\
    = & ~ \sum_{\ell=1}^{m} \sum_{i \neq j} (u_{\ell} u_{\ell}^{\top} u_i u_i^{\top}) \otimes (u_{\ell} u_{\ell}^{\top} u_j u_j^{\top}) \notag \\
    = & ~ 0, 
\end{align}
where the first step follows from $(A \otimes B) \cdot (C \otimes D) = (AC) \otimes (BD)$ , the second step follows that $u_i^{\top} u_j = 0$ if $i \neq j$ and $\ell \neq i$ or $\ell \neq j$ always holds in Eq.~\eqref{eq:u_ell_i_j_product}.

Therefore, we get that $Q_2 = 0$.
\end{proof}

\subsection{Convergence for multiple iterations}\label{sec:gd_analysis_2}

The goal of this section is to prove the convergence of Algorithm~\ref{alg:GD}:

\begin{lemma}[Convergence of gradient descent]\label{lem:gd_convergence}
Suppose the measurement vectors $\{u_i\}_{i\in [m]}$ are orthogonal unit vectors, and suppose $|b_i|$ is bounded by $R$ for $i\in [m]$.  Then, for any $\delta \in (0,1)$, if we take $\lambda = \Omega(\delta^{-1}\log m)$ and $\epsilon=O(\lambda^{-1})$ in Algorithm~\ref{alg:GD}, then for $T=\widetilde{\Omega}(\sqrt{m}R\delta^{-1})$ iterations, the solution matrix $A_T$ satisfies:
\begin{align*}
    | u_i^\top A_{T} u_i - b_i| \leq \delta~~~\forall i\in [m].
\end{align*}
\end{lemma}

\begin{proof}
We defer the proof to Appendix~\ref{sec:gd_convergence_proof}
\end{proof}

Theorem~\ref{thm:gd_orthogonal} follows immediately from Lemma~\ref{lem:gd_cost_per_iter} and Lemma~\ref{lem:gd_convergence}.
\section{Stochastic gradient descent}\label{sec:sgd}
In this section, we show that the cost-per-iteration of the approximate matrix sensing algorithm can be improved by using a stochastic gradient descent (SGD). More specifically, SGD can obtain a $\delta$-approximate solution with $O(Bn^2)$, where $0<B<m$ is the size of the mini batch in SGD. Below is the main theorem of this section:

\begin{theorem}[Stochastic gradient descent for orthogonal measurements]\label{thm:sgd_orthogonal}
Suppose $u_1,\dots,u_m\in \R^n$ are orthogonal unit vectors, and suppose $|b_i|\leq R$ for all $i\in [m]$. There exists an algorithm such that for any $\delta \in (0,1)$, performs $\wt{O}(m^{3/2}B^{-1}R\delta^{-1})$ iterations of gradient descent with $O(Bn^2)$-time per iteration and outputs a matrix $A\in \R^{n\times n}$ satisfies:
\begin{align*}
    | u_i^\top A u_i - b_i| \leq \delta~~~\forall i\in [m].
\end{align*}
\end{theorem}

The algorithm and its time complexity are provided in Section~\ref{sec:sgd_alg}. The convergence is proved in Section~\ref{sec:sgd_converge_1} and \ref{sec:sgd_converge_2}. The SGD algorithm for the general measurement without the assumption that the $\{u_i\}_{i\in [m]}$ are orthogonal vectors is deferred to Appendix~\ref{sec:sgd_general}.

\subsection{Algorithm}\label{sec:sgd_alg}
We can use the stochastic gradient descent algorithm (Algorithm~\ref{alg:stochastic_gradient_descent}) for matrix sensing. More specifically, in each iteration, we will uniformly sample a subset ${\cal B}\subset [m]$ of size $B$, and then compute the gradient of the stochastic potential function:
\begin{align}\label{eq:sgd_potential_func}
        \nabla \Phi_{\lambda}(A{ , {\cal B}}) := {  \frac{m}{|{\cal B}|}\sum_{i \in {\cal B}}} u_i u_i^\top \lambda \sinh( \lambda (u_i^\top A u_i - b_i) ), 
\end{align}
which is an $n$-by-$n$ matrix. Then, we do the following gradient step:
\begin{align}
    A_{t+1} \gets A_t - \epsilon \cdot \nabla \Phi_{\lambda}(A_t,{\cal B}_t) / \| \nabla \Phi_{\lambda}(A_t) \|_F.
\end{align}

\begin{lemma}[Running time of stochastic gradient descent]\label{lem:sgd_cost_per_iter}
Algorithm~\ref{alg:stochastic_gradient_descent} takes $O(mn^2)$-time for preprocessing and 
each iteration takes
$
    O(Bn^2)
$-time.
\end{lemma}

\begin{proof}
The time-consuming step is to compute $\|\nabla \Phi_\lambda(A_t)\|_F$. Since
\begin{align*}
    \nabla \Phi_{\lambda}(A_t) = \sum_{i=1}^m u_i u_i^\top \lambda \sinh\left( \lambda (u_i^\top A_t u_i - b_i)\right),
\end{align*}
and $u_i\bot u_j$ for $i\ne j\in [m]$, we know that $u_i$ is an eigenvector of $\nabla \Phi_{\lambda}(A)$ with eigenvalue $\lambda \sinh\left( \lambda (u_i^\top A_t u_i - b_i)\right)$ for each $i\in [m]$. Thus, we have
\begin{align*}
    \|\nabla \Phi_\lambda(A_t)\|_F^2 = &~ \sum_{i=1}^m \lambda^2 \sinh^2\left( \lambda (u_i^\top A_t u_i - b_i)\right)\\
    = &~ \sum_{i=1}^m \lambda^2 \sinh^2 (\lambda z_{t,i}),
\end{align*}
where $z_{t,i}:=u_i^\top A_t u_i - b_i$ for $i\in [m]$. Then, if we know ${z_{t,i}}_{i\in [m]}$, we can compute $\|\nabla \Phi_\lambda(A_t)\|_F$ in $O(m)$-time.

Consider the change $z_{t+1,i}-z_{t,i}$:
\begin{align*}
    &z_{t+1,i}-z_{t,i}\\
    = &~ u_i^\top (A_{t+1}-A_t)u_i\\
    = &~  -\frac{\epsilon}{\| \nabla \Phi_{\lambda}(A_t) \|_F}\cdot u_i^\top \nabla \Phi_{\lambda}(A_t,{\cal B} _t)u_i\\
    = &~ -\frac{\epsilon \lambda m }{\| \nabla \Phi_{\lambda}(A_t) \|_F B} \sum_{j\in {\cal B}_t}u_i^\top u_j u_j^\top u_i \cdot  \sinh(\lambda z_{t,j})\\
    = &~ -\frac{\epsilon \lambda m  \sinh(\lambda z_{t,i})}{\| \nabla \Phi_{\lambda}(A_t) \|_F B} \cdot {\bf 1}_{i\in {\cal B}_t},
\end{align*}
where the last step follows from $u_i\bot u_j$ for $i\ne j$.
Hence, if we have already computed $\{z_{t,i}\}_{i\in [m]}$ and $\|\nabla \Phi_\lambda(A_t)\|_F$, $\{z_{t+1,i}\}_{i\in [m]}$ can be obtained in $O(B)$-time.

Therefore, we preprocess $z_{1,i}=u_i^\top A_1u_i-b_i$ for all $i\in [m]$ in $O(mn^2)$-time. Then, in the $t$-th iteration ($t>0$), we first compute 
\begin{align*}
    \nabla \Phi_{\lambda} (A_t, {\cal B}_t) = \frac{m}{B}\sum_{i\in {\cal B}_t}u_iu_i^\top \lambda \sinh(\lambda z_{t,i})
\end{align*}
in $O(Bn^2)$-time. Next, we compute $\|\nabla \Phi_\lambda(A_t)\|_F$ using $z_{t,i}$ in $O(m)$-time. $A_{t+1}$ can be obtained in $O(n^2)$-time. Finally, we use $O(B)$-time to update $\{z_{t+1,i}\}_{i\in [m]}$.

Hence, the total running time per iteration is
\begin{align*}
    O(Bn^2 + m + n^2 + B) = O(Bn^2).
\end{align*}
\end{proof}

\begin{algorithm}[t]\caption{Matrix Sensing by Stochastic Gradient Descent.}\label{alg:stochastic_gradient_descent}
\begin{algorithmic}[1]
\Procedure{SGD}{$\{u_i,b_i\}_{i\in [m]}$} \Comment{Theorem~\ref{thm:sgd_orthogonal}}
    \State $\tau \gets \max_{i \in [m]} b_i $
    \State $A_1 \gets \tau \cdot I$
    \State $z_i\gets u_i^\top A_1 u_i - b_i$ for $i\in [m]$
    \For{$t = 1 \to T$}
        \State Sample ${\cal B}_t \subset [m]$ of size $B$ uniformly at random
        \State $\nabla \Phi_{\lambda} (A_t, {\cal B}_t) \gets \frac{m}{B} \sum_{i \in {\cal B}_t} u_i u_i^\top \lambda \sinh( \lambda z_{i} ) $
        \State $\| \nabla \Phi_{\lambda} (A_t) \|_F\gets \left(\sum_{i=1}^m \lambda^2 \sinh^2 (\lambda z_{i})\right)^{1/2}$
        \State $A_{t+1} \gets A_t - \epsilon \cdot \nabla \Phi_{\lambda}(A_t,{\cal B} _t) / \| \nabla \Phi_{\lambda}(A_t) \|_F$
        \For{$i\in {\cal B}_t$}
            \State \hspace{-2mm} $z_i\gets z_i - \epsilon  \lambda m \sinh(\lambda z_{i}) / (\| \nabla \Phi_{\lambda}(A_t) \|_F B)$
        \EndFor
    \EndFor
    \State \Return $A_{T+1}$
\EndProcedure
\end{algorithmic}
\end{algorithm}

\subsection{Analysis of One Iteration}\label{sec:sgd_converge_1}

Suppose $A \in \R^{n \times n}$. Let ${\cal B}_t$ be a uniformly random $B$-subset of $[m]$ at the $t$-th iteration, where $B$ is a parameter.

We can compute the gradient of $\Phi_{\lambda}(A{ , {\cal B}})$ with respect to $A$ as follows:
\begin{align*}
    & ~ \nabla \Phi_{\lambda}(A{ , {\cal B}}) \\
    = & ~ {  \frac{m}{|{\cal B}|}\sum_{i \in {\cal B}}} u_i u_i^\top \lambda \sinh( \lambda (u_i^\top A u_i - b_i) ), 
\end{align*}
where $\nabla \Phi_{\lambda}(A{ , {\cal B}})\in \R^{n\times n}$.

We can also compute the Hessian of $\Phi_{\lambda}(A{ , {\cal B}})$ with respect to $A$ as follows:
\begin{align*}
     & ~\nabla^2 \Phi_{\lambda}(A{ , {\cal B}}) \\
    = & ~ {  \frac{m}{|{\cal B}|}\sum_{i \in {\cal B}}} ( u_i u_i^\top ) \otimes ( u_i u_i^\top ) \lambda^2 \cosh( \lambda (u_i^\top A u_i - b_i) ) 
\end{align*}
where $\nabla^2 \Phi_{\lambda}(A{ , {\cal B}})\in \R^{n^2\times n^2}$ and $\otimes$ is the Kronecker product.

{It is easy to see the expectations of the gradient and Hessian of $\Phi_\lambda(A,{\cal B})$ over a random set ${\cal B}$:
\begin{align*}
    & ~\E_{{\cal B} \sim [m]}[\nabla \Phi_{\lambda}(A, {\cal B})] =    \nabla \Phi_{\lambda}(A), \\
    & ~\E_{{\cal B} \sim [m]}[\nabla^2 \Phi_{\lambda}(A, {\cal B})] =    \nabla^2 \Phi_{\lambda}(A)
\end{align*}
}

\begin{lemma}[Expected progress on potential]\label{lem:stochastic_gradient_descent}
Given $m$ vectors $u_1, u_2, \cdots, u_m \in \R^n$. 
Assume $\langle u_i , u_j \rangle = 0 $ for any $i \neq j \in [m]$ and $\| u_i\|^2 = 1$, for all $i \in [m]$.
Let $\epsilon \lambda \leq 0.01 \frac{|{\cal B}_t|}{m}$, for all $t>0$.

Then, we have
\begin{align*}
    \E[\Phi_{\lambda} ( A_{t+1} )] \leq (1-0.9 \frac{ \lambda \epsilon }{\sqrt{m} }) \cdot { \Phi_{\lambda} (A_t )} +  \lambda \epsilon \sqrt{m}.
\end{align*}
 
\end{lemma}

\begin{proof}
We first express the expectation as follows:
\begin{align}\label{eq:expectation_phi_update}
    & ~ \E_{A_{t+1}}[\Phi_{\lambda}(A_{t+1})] - \Phi_{\lambda} (A_t)\notag \\
    \leq & ~ \E_{A_{t+1}}[ \langle  \nabla \Phi_{\lambda} (A_t) ,  (A_{t+1} - A_t) \rangle ] \notag \\
    + & ~ O(1) \cdot \E_{A_{t+1}}[ \langle  \nabla^2 \Phi_{\lambda}(A_t) , (A_{t+1} - A_t) \otimes (A_{t+1} - A_t) \rangle ],
\end{align}
which follows from Corollary~\ref{cor:matrix_der_trace}.

We choose
\begin{align*}
    A_{t+1} = A_t - \epsilon \cdot \nabla \Phi_{\lambda}(A_t{ , {\cal B}_{t}}) / \| \nabla \Phi_{\lambda}(A_t) \|_F.
\end{align*}
Then, we can bound
\begin{align}\label{eq:sgd_tr_phi_At_first_moment} 
 & ~ \E_{A_{t+1}} [ -\tr [ \nabla \Phi_{\lambda} (A_t) \cdot (A_{t+1} - A_t) ] ] \notag \\
= & ~ \E_{{\cal B}_t}\Big[ \tr\Big[ \nabla \Phi_{\lambda}(A_t) \cdot \frac{\epsilon \nabla \Phi_{\lambda} (A_t,{\cal B}_t) }{ \| \nabla \Phi_{\lambda}(A_t) \|_F } \Big] \Big] \notag \\
= & ~ \epsilon \cdot \| \nabla \Phi_{\lambda} (A_t)  \|_F 
\end{align}

We define for $t>0$ and $i\in [m]$, 
\begin{align*}
    z_{t,i} := u_i^\top A_{t} u_i - b_i.
\end{align*}

We need to compute this $\Delta_2$. For simplificity, we consider $\Delta_2 \cdot \| \nabla \Phi_{\lambda} (A_t) \|_F^2$,
\begin{align}\label{eq:sgd_tr_phi_At_second_moment}
     = & ~  \tr[ \nabla^2 \Phi_{\lambda}(A_t) \cdot (A_{t+1} - A_t) \otimes (A_{t+1} - A_t) ] \cdot  \| \nabla \Phi_{\lambda}(A_t) \|_F^2 \notag\\
    = & ~ (\lambda \epsilon)^2 \cdot (\frac{m}{|{\cal B}_t|})^2  \cdot   \tr\Big[ \nabla^2 \Phi_{\lambda}(A_t) \cdot 
     ( \sum_{i\in {\cal B}_t} u_i u_i^\top  \sinh( z_{t,i} )    \otimes ( \sum_{i \in {\cal B}_t} u_i u_i^\top \sinh( z_{t,i} ) \Big].
\end{align}
Ignoring the scalar factor in the above equation, we have
\begin{align}
    = &~    \tr\Big[ \nabla^2 \Phi_{\lambda}(A_t) \cdot  ( \sum_{i,j \in B_t}  \sinh( z_{t,i} )\sinh( z_{t,i} ) \cdot   (u_i u_i^\top \otimes u_ju_j^\top) ) \Big] \notag\\
    = &    \tr\Big[ \nabla^2 \Phi_{\lambda}(A_t) \cdot ( \sum_{i \in B_t} \sinh^2( z_{t,i }) (u_i u_i^\top \otimes u_iu_i^\top) ) \Big] \notag\\
    + &    \tr\Big[ \nabla^2 \Phi_{\lambda}(A_t) \cdot ( \sum_{i\neq j \in B_t} \sinh( z_{t,i} )\sinh( z_{t,i} ) \cdot    (u_i u_i^\top \otimes u_j u_j^\top) ) \Big] \notag\\
    =: & ~    \wt{Q}_1 + \wt{Q}_2  ,
\end{align}
where the first step follows that we extract the scalar values from Kronecker product,
the second step comes from splitting into two partitions based on whether $i = j$, the third step comes from the definition of $\wt{Q}_1 $ and $\wt{Q}_2$ where $\wt{Q}_1$ denotes the diagonal term, and $\wt{Q}_2$ denotes the off-diagonal term.
Taking expectation, we have 
\begin{align}\label{eq:expectation_Delta_2}
    & ~ \E[ \Delta_2 \cdot \| \nabla \Phi_{\lambda} (A_t) \|_F^2 ] \notag\\
    = & ~ (\lambda \epsilon)^2 \cdot ( \frac{m}{|{\cal B}_t|} )^2 \E[ \wt{Q}_1] \notag\\
    = & ~ (\lambda \epsilon)^2 \cdot ( \frac{m}{|{\cal B}_t|} )^2 \cdot \frac{|{\cal B}_t|}{m} \cdot Q_1 \notag\\
    \leq & ~ (\lambda \epsilon)^2 \cdot \frac{m}{|{\cal B}_t|} 
      \cdot ( \sqrt{m} + \frac{1}{\lambda} \| \nabla \Phi_{\lambda}(A_t) \|_F ) \cdot \| \nabla \Phi_{\lambda} (A_t) \|_F^2
\end{align}
where the first step comes from extracting the constant terms from the expectation and Claim~\ref{cla:gd_Q2},
the second step follows that $\E[\wt{Q}_1] = \frac{|{\cal B}_t|}{m} \cdot Q_1$,
and the third step comes from the Claim~\ref{cla:gd_Q1}.
Therefore, we have:
\begin{align*}
    & ~ \E[ \Phi_{\lambda} (A_{t+1}) ] - \Phi_{\lambda} (A_t) \\
    \leq &  - \E[ \Delta_1 ]   + O(1) \cdot \E[ \Delta_2 ] \\ 
    \leq &  - \epsilon (1 - O(\epsilon \lambda) \cdot \frac{m}{|{\cal B}_t|} ) \| \nabla \Phi_{\lambda}(A_t) \|_F+ O(\epsilon \lambda)^2 \sqrt{m} \\
    \leq &  - 0.9 \epsilon  \| \nabla \Phi_{\lambda}(A_t) \|_F+ 
    O(\epsilon \lambda)^2 \sqrt{m} \\
    \leq &  -0.9 \epsilon \lambda \frac{1}{\sqrt{m}} ( \Phi_{\lambda}(A_t) - m ) +O(\epsilon \lambda)^2 \sqrt{m} \\
    \leq &   -0.9 \epsilon \lambda \frac{1}{\sqrt{m}} \Phi_{\lambda}(A_t) + \epsilon \lambda \sqrt{m},
\end{align*}
where the first step comes from Eq.~\eqref{eq:expectation_phi_update},
the second step comes from Eq.~\eqref{eq:sgd_tr_phi_At_first_moment} and Eq.~\eqref{eq:expectation_Delta_2},
the third step follows from $\epsilon \leq 0.01 \frac{|{\cal B}_t|}{\lambda m}$,
the forth step follows from Eq.~\eqref{eq:phi_At_F_norm},
and the last step follows from $\epsilon \lambda \in (0, 0.01)$.
\end{proof}

\subsection{Convergence for multiple iterations}\label{sec:sgd_converge_2}
The goal of this section is to prove the convergence of Algorithm~\ref{alg:stochastic_gradient_descent}. 
\begin{lemma}[Convergence of stochastic gradient descent]\label{lem:sgd_convergence}
Suppose the measurement vectors $\{u_i\}_{i\in [m]}$ are orthogonal unit vectors, and suppose $|b_i|$ is bounded by $R$ for $i\in [m]$. Then, for any $\delta \in (0,1)$, if we take $\lambda = \Omega(\delta^{-1}\log m)$ and $\epsilon=O(\lambda^{-1}m^{-1}B)$ in Algorithm~\ref{alg:stochastic_gradient_descent}, then for 
\begin{align*} 
T=\widetilde{\Omega}(m^{3/2}B^{-1}R\delta^{-1})
\end{align*}
iterations, with high probability, the solution matrix $A_T$ satisfies:
\begin{align*}
    | u_i^\top A_{T+1} u_i - b_i| \leq \delta~~~\forall i\in [m].
\end{align*}
\end{lemma}

\begin{proof}
Similar to the proof of Lemma~\ref{lem:gd_convergence}, we can bound the initial potential by:
\begin{align*}
    \Phi(A_1) \leq 2^{O(\lambda R)}.
\end{align*}

In the following iterations, by Lemma~\ref{lem:stochastic_gradient_descent}, we have
\begin{align*}
    \E[\Phi_{\lambda} ( A_{t+1} )] \leq (1-0.9 \frac{ \lambda \epsilon }{\sqrt{m} }) \cdot { \Phi_{\lambda} (A_t )} +  \lambda \epsilon \sqrt{m},
\end{align*}
as long as $\epsilon \leq 0.01\frac{|B_t|}{\lambda m}$, where $B_t$ is a uniformly random subset of $[m]$ of size $B$.

It suffices to take $\epsilon = O(\lambda^{-1} m^{-1}B)$.

Now, we can apply Lemma~\ref{lem:stochastic_gradient_descent} for $T$ times and get that
\begin{align*}
    \E[\Phi(A_{T+1})]
    \leq 2^{-\Omega( T  \epsilon \lambda / \sqrt{m} ) + O(\lambda R)} + 2  m.
\end{align*}
By taking 
\begin{align*}
T=\widetilde{\Omega}(m^{3/2}B^{-1}R\delta^{-1}),
\end{align*}
we have
\begin{align*}
    \Phi(A_{T+1}) \leq O(m)
\end{align*}
holds with high probability. By the same argument as in the proof of Lemma~\ref{lem:gd_convergence}, we have
\begin{align*}
    | u_i^\top A_{T+1} u_i - b_i| \leq \delta~\forall i\in [m].
\end{align*}
The lemma is thus proved.
\end{proof}
\ifdefined\isarxiv
\bibliographystyle{alpha}
\bibliography{ref}
\else
\bibliography{ref}

\newcommand{\etalchar}[1]{$^{#1}$}
\begin{thebibliography}{CLMW11}

\bibitem[Aar07]{a07}
Scott Aaronson.
\newblock The learnability of quantum states.
\newblock {\em Proceedings of the Royal Society A: Mathematical, Physical and
  Engineering Sciences}, 463(2088):3089--3114, 2007.

\bibitem[BLSS20]{blss20}
Jan van~den Brand, Yin~Tat Lee, Aaron Sidford, and Zhao Song.
\newblock Solving tall dense linear programs in nearly linear time.
\newblock In {\em Proceedings of the 52nd Annual ACM SIGACT Symposium on Theory
  of Computing (STOC)}, pages 775--788, 2020.

\bibitem[Bra20]{b20}
Jan van~den Brand.
\newblock A deterministic linear program solver in current matrix
  multiplication time.
\newblock In {\em Proceedings of the Fourteenth Annual ACM-SIAM Symposium on
  Discrete Algorithms (SODA)}, pages 259--278. SIAM, 2020.

\bibitem[CLMW11]{clmw11}
Emmanuel~J. Cand\`{e}s, Xiaodong Li, Yi~Ma, and John Wright.
\newblock Robust principal component analysis?
\newblock {\em J. ACM}, 58(3), 2011.

\bibitem[CLS19]{cls19}
Michael~B. Cohen, Yin~Tat Lee, and Zhao Song.
\newblock Solving linear programs in the current matrix multiplication time.
\newblock In {\em Proceedings of the 51st Annual ACM SIGACT Symposium on Theory
  of Computing}, STOC 2019, page 938–942, New York, NY, USA, 2019.
  Association for Computing Machinery.

\bibitem[CZ15]{cz15}
T~Tony Cai and Anru Zhang.
\newblock Rop: Matrix recovery via rank-one projections.
\newblock {\em The Annals of Statistics}, 43(1):102--138, 2015.

\bibitem[DLS23]{dls23}
Yichuan Deng, Zhihang Li, and Zhao Song.
\newblock An improved sample complexity for rank-1 matrix sensing.
\newblock {\em arXiv preprint arXiv:2303.06895}, 2023.

\bibitem[DLY21]{dly21}
Sally Dong, Yin~Tat Lee, and Guanghao Ye.
\newblock A nearly-linear time algorithm for linear programs with small
  treewidth: A multiscale representation of robust central path.
\newblock In {\em Proceedings of the 53rd Annual ACM SIGACT Symposium on Theory
  of Computing}, pages 1784--1797, 2021.

\bibitem[FGLE12]{fgle12}
Steven~T Flammia, David Gross, Yi-Kai Liu, and Jens Eisert.
\newblock Quantum tomography via compressed sensing: error bounds, sample
  complexity and efficient estimators.
\newblock {\em New Journal of Physics}, 14(9):095022, 2012.

\bibitem[GS22]{gs22}
Yuzhou Gu and Zhao Song.
\newblock A faster small treewidth sdp solver.
\newblock {\em arXiv preprint arXiv:2211.06033}, 2022.

\bibitem[HJS{\etalchar{+}}22]{hjstz21}
Baihe Huang, Shunhua Jiang, Zhao Song, Runzhou Tao, and Ruizhe Zhang.
\newblock Solving sdp faster: A robust ipm framework and efficient
  implementation.
\newblock In {\em FOCS}, 2022.

\bibitem[HRVW96]{hrvw96}
Christoph Helmberg, Franz Rendl, Robert~J Vanderbei, and Henry Wolkowicz.
\newblock An interior-point method for semidefinite programming.
\newblock {\em SIAM Journal on optimization}, 6(2):342--361, 1996.

\bibitem[HSWZ22]{hswz22}
Hang Hu, Zhao Song, Omri Weinstein, and Danyang Zhuo.
\newblock Training overparametrized neural networks in sublinear time.
\newblock {\em arXiv preprint arXiv:2208.04508}, 2022.

\bibitem[JKL{\etalchar{+}}20]{jklps20}
Haotian Jiang, Tarun Kathuria, Yin~Tat Lee, Swati Padmanabhan, and Zhao Song.
\newblock A faster interior point method for semidefinite programming.
\newblock In {\em 2020 IEEE 61st annual symposium on foundations of computer
  science (FOCS)}, pages 910--918. IEEE, 2020.

\bibitem[JLSW20]{jlsw20}
Haotian Jiang, Yin~Tat Lee, Zhao Song, and Sam Chiu-wai Wong.
\newblock An improved cutting plane method for convex optimization,
  convex-concave games, and its applications.
\newblock In {\em Proceedings of the 52nd Annual ACM SIGACT Symposium on Theory
  of Computing}, pages 944--953, 2020.

\bibitem[JM13]{jm13}
Adel Javanmard and Andrea Montanari.
\newblock Localization from incomplete noisy distance measurements.
\newblock {\em Found. Comput. Math.}, 13(3):297–345, jun 2013.

\bibitem[JMD10]{jmd10}
Prateek Jain, Raghu Meka, and Inderjit Dhillon.
\newblock Guaranteed rank minimization via singular value projection.
\newblock {\em Advances in Neural Information Processing Systems}, 23, 2010.

\bibitem[JN08]{jn08}
Anatoli Juditsky and Arkadii~S Nemirovski.
\newblock Large deviations of vector-valued martingales in 2-smooth normed
  spaces.
\newblock {\em arXiv preprint arXiv:0809.0813}, 2008.

\bibitem[JSWZ21]{jswz21}
Shunhua Jiang, Zhao Song, Omri Weinstein, and Hengjie Zhang.
\newblock Faster dynamic matrix inverse for faster lps.
\newblock In {\em STOC}. arXiv preprint arXiv:2004.07470, 2021.

\bibitem[KKD15]{kkd15}
Amir Kalev, Robert~L Kosut, and Ivan~H Deutsch.
\newblock Quantum tomography protocols with positivity are compressed sensing
  protocols.
\newblock {\em npj Quantum Information}, 1(1):1--6, 2015.

\bibitem[KRT17]{krt17}
Richard Kueng, Holger Rauhut, and Ulrich Terstiege.
\newblock Low rank matrix recovery from rank one measurements.
\newblock {\em Applied and Computational Harmonic Analysis}, 42(1):88--116,
  2017.

\bibitem[LB09]{lb09}
Kiryung Lee and Yoram Bresler.
\newblock Guaranteed minimum rank approximation from linear observations by
  nuclear norm minimization with an ellipsoidal constraint.
\newblock {\em arXiv preprint arXiv:0903.4742}, 2009.

\bibitem[Liu11]{l11}
Yi-Kai Liu.
\newblock Universal low-rank matrix recovery from pauli measurements.
\newblock {\em Advances in Neural Information Processing Systems}, 24, 2011.

\bibitem[LMCC19]{lmcc19}
Yuanxin Li, Cong Ma, Yuxin Chen, and Yuejie Chi.
\newblock Nonconvex matrix factorization from rank-one measurements.
\newblock In {\em The 22nd International Conference on Artificial Intelligence
  and Statistics}, pages 1496--1505. PMLR, 2019.

\bibitem[LS14]{ls14}
Yin~Tat Lee and Aaron Sidford.
\newblock Path finding methods for linear programming: Solving linear programs
  in o (vrank) iterations and faster algorithms for maximum flow.
\newblock In {\em 2014 IEEE 55th Annual Symposium on Foundations of Computer
  Science}, pages 424--433. IEEE, 2014.

\bibitem[LS15]{ls15}
Yin~Tat Lee and Aaron Sidford.
\newblock Efficient inverse maintenance and faster algorithms for linear
  programming.
\newblock In {\em 2015 IEEE 56th Annual Symposium on Foundations of Computer
  Science}, pages 230--249. IEEE, 2015.

\bibitem[LSW15]{lsw15}
Yin~Tat Lee, Aaron Sidford, and Sam Chiu-wai Wong.
\newblock A faster cutting plane method and its implications for combinatorial
  and convex optimization.
\newblock In {\em 2015 IEEE 56th Annual Symposium on Foundations of Computer
  Science}, pages 1049--1065. IEEE, 2015.

\bibitem[LSZ19]{lsz19}
Yin~Tat Lee, Zhao Song, and Qiuyi Zhang.
\newblock Solving empirical risk minimization in the current matrix
  multiplication time.
\newblock In {\em COLT}, 2019.

\bibitem[LV10]{lv10}
Zhang Liu and Lieven Vandenberghe.
\newblock Interior-point method for nuclear norm approximation with application
  to system identification.
\newblock {\em SIAM Journal on Matrix Analysis and Applications},
  31(3):1235--1256, 2010.

\bibitem[NN94]{nn94}
Yurii Nesterov and Arkadii Nemirovskii.
\newblock {\em Interior-point polynomial algorithms in convex programming}.
\newblock SIAM, 1994.

\bibitem[QSZZ23]{qszz23}
Liank Qin, Zhao Song, Lichen Zhang, and Danyang Zhuo.
\newblock An online and unified algorithm for projection matrix vector
  multiplication with application to empirical risk minimization.
\newblock In {\em AISTATS}, 2023.

\bibitem[RFP10]{rfp10}
Benjamin Recht, Maryam Fazel, and Pablo~A Parrilo.
\newblock Guaranteed minimum-rank solutions of linear matrix equations via
  nuclear norm minimization.
\newblock {\em SIAM review}, 52(3):471--501, 2010.

\bibitem[SY21]{sy21}
Zhao Song and Zheng Yu.
\newblock Oblivious sketching-based central path method for solving linear
  programming problems.
\newblock In {\em 38th International Conference on Machine Learning (ICML)},
  2021.

\bibitem[SYZ21]{syz21}
Zhao Song, Shuo Yang, and Ruizhe Zhang.
\newblock Does preprocessing help training over-parameterized neural networks?
\newblock {\em Advances in Neural Information Processing Systems}, 34, 2021.

\bibitem[SZZ21]{szz21}
Zhao Song, Lichen Zhang, and Ruizhe Zhang.
\newblock Training multi-layer over-parametrized neural network in subquadratic
  time.
\newblock {\em arXiv preprint arXiv:2112.07628}, 2021.

\bibitem[TBS{\etalchar{+}}16]{tbssr16}
Stephen Tu, Ross Boczar, Max Simchowitz, Mahdi Soltanolkotabi, and Ben Recht.
\newblock Low-rank solutions of linear matrix equations via procrustes flow.
\newblock In {\em International Conference on Machine Learning}, pages
  964--973. PMLR, 2016.

\bibitem[WSB11]{wsb11}
Andrew Waters, Aswin Sankaranarayanan, and Richard Baraniuk.
\newblock Sparcs: Recovering low-rank and sparse matrices from compressive
  measurements.
\newblock {\em Advances in neural information processing systems}, 24, 2011.

\bibitem[WZG17]{wzg17}
Lingxiao Wang, Xiao Zhang, and Quanquan Gu.
\newblock A unified computational and statistical framework for nonconvex
  low-rank matrix estimation.
\newblock In {\em Artificial Intelligence and Statistics}, pages 981--990.
  PMLR, 2017.

\bibitem[ZJD15]{zjd15}
Kai Zhong, Prateek Jain, and Inderjit~S Dhillon.
\newblock Efficient matrix sensing using rank-1 gaussian measurements.
\newblock In {\em International conference on algorithmic learning theory},
  pages 3--18. Springer, 2015.

\end{thebibliography}
\bibliographystyle{icml2023}
\fi

\appendix
\onecolumn
\section*{Appendix}
\paragraph{Roadmap.} We first provide the proofs for matrix hyperbolic functions and properties of $\sinh$ and $\cosh$ in Appendix~\ref{sec:preli_app}. Then we provide the  proofs for the gradient descent and stochastic gradient descent convergence analysis in Appendix~\ref{sec:gd_sgd_missing_proofs}. We consider the spectral potential function with ground-truth oracle scenario in Appendix~\ref{sec:potential_ground_truth_app}. We analyze the gradient descent with non-orthogonal measurements in Appendix~\ref{sec:gd_non_orthogonal}. We provide the cost-per-iteration analysis for stochastic gradient descent under non-orthogonal measurements in Appendix~\ref{sec:sgd_general}.

\section{Proofs of Preliminary Lemmas}\label{sec:preli_app}

\subsection{Calculus tools}
We state a useful calculus tool from prior work,
\begin{lemma}[Proposition 3.1 in \cite{jn08}]\label{lem:jn08}
Let $\Delta$ be an open interval on the axis, and $f$ be $C^2$ function on $\Delta$ such that for certain $\theta_{\pm}, \mu_{\pm} \in \R$ one has
\begin{align*}
    & ~ \forall (a<b, a,b \in \Delta) : \\
    & ~ \theta_- \cdot \frac{ f''(a) + f''(b) }{2} + \mu_- \leq \frac{f'(b) - f'(a)}{b-a} \\
    & ~ \frac{f'(b) - f'(a)}{b-a} \leq \theta_+ \cdot \frac{f''(a) + f''(b)}{2} + \mu_+,
\end{align*}
where $f'$ and $f''$ means the first- and second-order derivatives of $f$, respectively.

Let, further, ${\cal X}_n (\Delta)$ be the set of all $n \times n$ symmetric matrices with eigenvalues belonging to $\Delta$. Then ${\cal X}_n(\Delta)$ is an open convex set in the space $S^n$ of $n \times n$ symmetric matrices, the function
\begin{align*}
    F(X) = \tr [ f(X) ] : {\cal X}_n (\Delta) \rightarrow \R  
\end{align*}
is $C^2$, and for every $X \in {\cal X}_n (\Delta)$ and every $H \in S^n$ one has
\begin{align*}
    & ~ \theta_- \cdot \tr[ H f''(X) H ] + \mu_- \cdot \tr[ H^2 ] \leq D^2 F(X) [H,H] \\
   & ~ D^2 F(X) [H,H] \leq \theta_+ \cdot \tr[ H f''(X) H ] + \mu_+ \cdot \tr[H^2],
\end{align*}
where $D$ means directional derivative.
\end{lemma}

We will use below corollary to compute the trace with a map $f: \R\rightarrow \R $.
\begin{corollary}\label{cor:matrix_der_trace}
Let $f:\R\rightarrow \R$ be a $C^2$ function. Let $A$ and $B$ be two symmetric matrices. We have
\begin{align*}
    \tr[ f(A) ] 
   \leq & ~ \tr[ f(B) ] + \tr[ f'(B) ( A - B ) ] \\
    & ~ + O(1) \cdot \tr[ f''(B) ( A - B )^2 ]. 
\end{align*}
\end{corollary}

\subsection{Kronecker product}
Suppose we have two matrice $A \in \R^{m \times n}$ and $B \in \R^{p \times q}$, we use $A \otimes B$ denote the Kronecker product:
\begin{align*}
    A \otimes B = \left[\begin{array}{ccc}
A_{1,1} {B} & \cdots & A_{1,n} {B} \\
\vdots & \ddots & \vdots \\
A_{m,1} {B} & \cdots & A_{m,n} {B}
\end{array}\right].
\end{align*}

We state a fact and delay the proof into Section~\ref{sec:preli_app}.
\begin{fact}\label{fac:kronecker_product}
Suppose we have two matrices $A \in \R^{m \times n}$ and $B \in \R^{n \times k}$, we have
\begin{align*}
    (A \otimes B) \cdot (C \otimes D) = (AC) \otimes (BD).
\end{align*}
\end{fact}

\begin{proof}
From the definition of Kronecker product we have:
\begin{align*}
    &~ (A \otimes B) \cdot (C \otimes D) \\
    = &~ {\left[\begin{array}{ccc}
A_{1,1} B & \ldots & A_{1,n} B \\
\vdots & \ddots & \vdots \\
A_{m,1} B & \ldots & A_{m,n} B
\end{array}\right]\left[\begin{array}{ccc}
C_{1,1} D & \ldots & C_{1,k} D \\
\vdots & \ddots & \vdots \\
C_{n,1} D & \ldots & C_{n,k} D
\end{array}\right] } \\
   = &~ \left[\begin{array}{ccc}
(\sum_{i=1}^{n}A_{1,i}C_{i,1}) {BD} & \cdots & (\sum_{i=1}^{n}A_{1,i}C_{i,k}) {BD} \\
\vdots & \ddots & \vdots \\
(\sum_{i=1}^{n}A_{m,i}C_{i,1}) {BD} & \cdots & (\sum_{i=1}^{n}A_{m,i}C_{i,k}){BD}
\end{array}\right] \\
= &~ \left[\begin{array}{ccc}
(AC)_{1,1} BD & \cdots & (AC)_{1,k} BD\\
\vdots & \ddots & \vdots \\
(AC)_{m,1} BD & \cdots & (AC)_{m,k}BD
\end{array}\right] \\
= &~ (AC) \otimes (BD)
\end{align*}
Thus we complete the proof.
\end{proof}

\subsection{Proof of \texorpdfstring{$\cosh(A)$}{} upper bound}\label{sec:cosh_bound_proof}

\begin{lemma}[Restatement of Lemma~\ref{lem:cosh_bound}]\label{lem:cosh_bound_app}
Let $A$ be a real symmetric matrix, then we have
\begin{align*}
    \|\cosh(A)\| = \cosh(\|A\|) \leq \tr[\cosh(A)].
\end{align*}
We also have $\|A\| \leq 1+\log(\tr[\cosh(A)])$.
\end{lemma}

\begin{proof}
Note that for each eigenvalue $\lambda$ of $A$, we know that it corresponds to $\cosh(\lambda)$ for $\cosh(A)$. The second inequality follows from the fact that $\cosh(A)$ is psd.

For the second part, we know that $\exp(x)/2\leq \cosh(x)$, hence, $\exp(\|A\|)/2\leq \cosh(\|A\|)$, and
\begin{align*}
    \|A\| = & ~ \log(\exp(\|A\|)) \\
    \leq & ~ \log(2\cosh(\|A\|)) \\
    \leq & ~ 1+\log(\tr[\cosh(A)]),
\end{align*}
where the second step is by the monotonicity of $\log(\cdot)$ and $\exp(\|A\|)\leq 2\cosh(\|A\|)$, the last step is by $\cosh(\|A\|)\leq \tr[\cosh(A)]$.
\end{proof}
 
We state a fact as follows:  
\begin{fact}\label{fact:cosh_sinh_1}
For any real number $x$, $ \cosh^2(x) - \sinh^2(x) =1$
\end{fact}
From the definition of $\cosh(x)$ and $\sinh(x)$ we have:
\begin{align*}
      &~ \cosh^2(x) - \sinh^2(x) \\
    = &~ \frac{1}{4}(e^{2x} + 2 + e^{-2x}) - \frac{1}{4}(e^{2x} - 2 + e^{-2x}) \\
    = &~ 1
\end{align*}

We also have the following lemma for matrix. 
\begin{lemma}\label{lem:cosh_sinh}
Let $A$ be a real symmetric matrix, then we have
\begin{align*}
    \cosh^2(A) - \sinh^2(A) = & ~ I.
\end{align*}
\end{lemma}
\begin{proof}
Since $A$ is real symmetric, we write it in the eigendecomposition form: $A=U\Lambda U^\top$, then
\begin{align*}
    & ~ \cosh^2(A)-\sinh^2(A) \\
    = & ~ U\cosh^2(\Lambda)U^\top-U\sinh^2(\Lambda)U^\top \\
    = & ~ U(\cosh^2(\Lambda)-\sinh^2(\Lambda))U^\top \\
    = & ~ UU^\top \\
    = & ~ I,
\end{align*}
where the first step follows from $\cosh$ and $\sinh$ can be expressed as $\exp$, the third step is by applying entrywise the identity $\cosh^2(x)-\sinh^2(x)=1$.
\end{proof}

\section{Proofs of GD and SGD convergence}\label{sec:gd_sgd_missing_proofs}
In this section, we provide proofs of convergence analysis the gradient descent and stochastic gradient descent matrix sensing algorithms.
\subsection{Proof of GD Progress on Potential Function}\label{sec:proof_gradient_descent}
We start with the progress of the gradient on the potential function in below lemma.
\begin{lemma}[Restatement of Lemma~\ref{lem:gradient_descent}]
Assume that $u_i \perp u_j = 0 $ for any $i,j \in [m]$ and $\| u_i\|^2 = 1$. Let $c \in (0,1)$ denote a sufficiently small positive constant. Then, for any $\epsilon,\lambda>0$ such that $\epsilon\lambda \leq c$,

we have for any $t>0$,
\begin{align*}
    \Phi_{\lambda} ( A_{t+1} ) \leq (1-0.9 \frac{ \lambda \epsilon }{\sqrt{m} }) \cdot \Phi_{\lambda} (A_t) +  \lambda \epsilon \sqrt{m}
\end{align*}

\end{lemma}

\begin{proof}
We first Taylor expand $\Phi_\lambda(A_{t+1})$ as follows:
\begin{align}\label{eq:gd_define_Delta_1_and_Delta_2}
    & ~ \Phi_{\lambda}(A_{t+1})- \Phi_{\lambda} (A_t) \notag \\
    \leq & ~ \langle \nabla \Phi_{\lambda} (A_t) , (A_{t+1} - A_t) \rangle  + O(1) \langle \nabla^2 \Phi_{\lambda}(A_t), (A_{t+1} - A_t) \otimes (A_{t+1} - A_t) \rangle \notag \\
    := & ~ \Delta_1 + O(1) \cdot \Delta_2,
\end{align}
which follows from Lemma~\ref{lem:jn08}.

We choose
\begin{align*}
    A_{t+1} = A_t - \epsilon \cdot \nabla \Phi_{\lambda}(A_t) / \| \nabla \Phi_{\lambda}(A_t) \|_F.
\end{align*}

We can bound
\begin{align}\label{eq:tr_phi_At_first_moment} 
 \Delta_1 = & ~ \tr[\nabla \Phi_{\lambda}(A_t) (A_{t+1} - A_t)] \notag \\
=  & ~ -\epsilon \cdot \|\nabla \Phi_{\lambda} (A_t) \|_F.
\end{align}

Next, we upper-bound $\Delta_2$. Define
\begin{align*}
    z_{t,i} := \lambda (u_i^\top A_t u_i - b_i).
\end{align*}
and consider $\Delta_2 \cdot (\lambda\epsilon )^{-2} \cdot \| \nabla \Phi_{\lambda} (A_t) \|_F^2$, which can be expressed as:
\begin{align}\label{eq:tr_phi_At_second_moment}
&\Delta_2 \cdot (\lambda\epsilon )^{-2} \cdot \| \nabla \Phi_{\lambda} (A_t) \|_F^2\notag\\
     = & ~ (\lambda\epsilon )^{-2} \tr[ \nabla^2 \Phi_{\lambda}(A_t) \cdot (A_{t+1} - A_t) \otimes (A_{t+1} - A_t) ] \cdot 
      \| \nabla \Phi_{\lambda}(A_t) \|_F^2 \notag\\
    = & ~  \tr\Big[ \nabla^2 \Phi_{\lambda}(A_t) \cdot 
    ( \sum_{i=1}^m u_i u_i^\top  \sinh( z_{t,i} ) ) \otimes ( \sum_{i=1}^m u_i u_i^\top \sinh( z_{t,i} ) )\Big] \notag\\
    = & ~ \tr\Big[ \nabla^2 \Phi_{\lambda}(A_t) \cdot ( \sum_{i,j}  \sinh( z_{t,i} )\sinh( z_{t,i} ) (u_i u_i^\top \otimes u_ju_j^\top) ) \Big] \notag\\
    = & ~ \tr\Big[ \nabla^2 \Phi_{\lambda}(A_t) \cdot ( \sum_{i=1}^m \sinh^2( z_{t,i} )) (u_i u_i^\top \otimes u_iu_i^\top) ) \Big] \notag\\
    + & ~  \tr\Big[ \nabla^2 \Phi_{\lambda}(A_t) \cdot ( \sum_{i\neq j} \sinh( z_{t,i} )\sinh( z_{t,j} ) (u_i u_i^\top \otimes u_j u_j^\top) ) \Big] \notag\\
    =: & ~  Q_1 + Q_2  ,
\end{align}
where 
\begin{align}\label{eq:def_Q1}
Q_1:= \tr\Big[ \nabla^2 \Phi_{\lambda}(A_t) \cdot ( \sum_{i=1}^m \sinh^2( z_{t,i} )) (u_i u_i^\top \otimes u_iu_i^\top) ) \Big]
\end{align}
denotes the diagonal term, and
\begin{align}\label{eq:def_Q2}
    Q_2:=&~\tr\Big[ \nabla^2 \Phi_{\lambda}(A_t) \cdot
    ( \sum_{i\neq j} \sinh( z_{t,i} )\sinh( z_{t,j} ) (u_i u_i^\top \otimes u_j u_j^\top) ) \Big]
\end{align}
denotes the off-diagonal term. 
The first step comes from the definition of $\Delta_2$,
the second step follows fromr eplacing $A_{t+1} - A_t$ using Eq~\eqref{eq:A_t1_update}, the third step follows that we extract the scalar values from Kronecker product, the fourth step comes from splitting into two partitions based on whether $i = j$, the fifth step comes from the definition of $Q_1$ and $Q_2$.

Thus,
\begin{align}\label{eq:bound_gd_Delta_2}
    \Delta_2 = & ~ (\epsilon \lambda)^2 (Q_1 + Q_2) / \| \nabla \Phi_{\lambda}(A_t) \|_F^2 \notag \\
    = & ~ (\epsilon \lambda)^2 (Q_1 + 0) / \| \nabla \Phi_{\lambda}(A_t) \|_F^2 \notag \\
    = & ~ (\epsilon \lambda )^2 \cdot ( \sqrt{m} + \frac{1}{\lambda} \| \nabla \Phi_{\lambda}(A_t) \|_F ).
\end{align}
where the second step follows from Claim~\ref{cla:gd_Q2}, and the third step follows from Claim~\ref{cla:gd_Q1}.

Hence, we have
\begin{align*}
    & ~ \Phi_{\lambda} (A_{t+1}) - \Phi_{\lambda} (A_t) \\
    \leq & ~  \Delta_1    + O(1) \cdot \Delta_2 \\
    \leq & ~ - \epsilon \| \nabla \Phi_{\lambda}(A_t) \|_F   + O(1) (\epsilon \lambda)^2 (\sqrt{m} + \frac{1}{\lambda} \| \nabla \Phi_{\lambda}(A_t) \|_F ) \\
    \leq & ~ - 0.9 \epsilon \| \Phi_{\lambda}(A_t) \|_F+ O(\epsilon \lambda)^2 \sqrt{m} 
\end{align*}
where the first step follows from Eq.~\eqref{eq:gd_define_Delta_1_and_Delta_2}, the second step follows from Eq.~\eqref{eq:bound_gd_Delta_1} and Eq.~\eqref{eq:bound_gd_Delta_2},  the third step follows from $\epsilon \lambda \in (0, 0.01)$.

For $\| \Phi_{\lambda} (A_t) \|_F$, we have
\begin{align}\label{eq:phi_At_F_norm}
& ~ \frac{1}{\lambda^2} \| \nabla \Phi_{\lambda} (A_t) \|_F^2 \notag \\
= & ~   \tr[ (\sum_{i=1}^m u_i u_i^\top \sinh( \lambda( u_i^\top A_t u_i - b_i ) )  )^2 ] \notag\\
= & ~   \tr[ \sum_{i=1}^m (u_i u_i^\top)^{2} \sinh^2 ( \lambda( u_i^\top A_t u_i - b_i ) )  ] \notag\\ 
= & ~ \sum_{i=1}^m \sinh^2 ( \lambda ( u_i^\top A_t u_i - b_i ) ) \notag \\ 
\geq & ~ \frac{1}{m} ( \sum_{i=1}^m \cosh ( \lambda ( u_i^\top A_t u_i - b_i ) ) - m )^2 \notag \\
= & ~ \frac{1}{m} ( \Phi_{\lambda} (A_t) - m )^2,
\end{align}
where the first step comes from Eq.~\eqref{eq:gradient_phi_A},
the second  steps follow from $u_i^\top u_j = 0$, the third step follows from $\| u_i \|_2 = 1$, the forth step follows from Part 2 in Lemma~\ref{lem:property_sinh_cosh_scalar}, {  the fifth step follows from the definition of $\Phi_{\lambda}(A)$.}

Thus, we get that
\begin{align}\label{eq:bound_gd_Delta_1}
\| \Phi_{\lambda} (A_t) \|_F^2 \geq ~
\lambda \epsilon \cdot \frac{1}{ \sqrt{m} } | \Phi_{\lambda}(A_t) - m |,
\end{align}

It implies that
\begin{align*}
&\Phi_{\lambda} (A_{t+1}) - \Phi_{\lambda} (A_t)\\
\leq & ~ -0.9 \epsilon \lambda \frac{1}{\sqrt{m}} | \Phi_{\lambda}(A_t) - m| +O(\epsilon \lambda)^2 \sqrt{m}\\
\leq &~ -0.9 \epsilon \lambda \frac{1}{\sqrt{m}} | \Phi_{\lambda}(A_t) - m| + 0.1\epsilon \lambda\sqrt{m},
\end{align*}
where the second step follows from extracting the constant term from the summation.

Then, when $\Phi(A_t)> m$, we have
\begin{align*}
    \Phi_{\lambda} (A_{t+1}) \leq (1-0.9 \frac{ \lambda \epsilon }{\sqrt{m} }) \cdot \Phi_{\lambda} (A_t) +  \lambda \epsilon \sqrt{m}.
\end{align*}
When $\Phi(A_t)\leq m$, we have
\begin{align*}
    \Phi_{\lambda} (A_{t+1}) \leq (1+0.9 \frac{ \lambda \epsilon }{\sqrt{m} }) \cdot \Phi_{\lambda} (A_t) -0.8  \lambda \epsilon \sqrt{m}.
\end{align*}

The lemma is then proved.
\end{proof}

\subsection{Proof of GD Convergence}\label{sec:gd_convergence_proof}
In this section, we provide proofs of convergence analysis of gradient descent matrix sensing algorithm.

\begin{lemma}[Restatement of Lemma~\ref{lem:gd_convergence}]
Suppose the measurement vectors $\{u_i\}_{i\in [m]}$ are orthogonal unit vectors, and suppose $|b_i|$ is bounded by $R$ for $i\in [m]$.  Then, for any $\delta \in (0,1)$, if we take $\lambda = \Omega(\delta^{-1}\log m)$ and $\epsilon=O(\lambda^{-1})$ in Algorithm~\ref{alg:GD}, then for $T=\widetilde{\Omega}(\sqrt{m}R\delta^{-1})$ iterations, the solution matrix $A_T$ satisfies:
\begin{align*}
    | u_i^\top A_{T} u_i - b_i| \leq \delta~~~\forall i\in [m].
\end{align*}
\end{lemma}

\begin{proof}
Let $\tau = \max_{i\in [m]} b_i$. At the beginning, we choose the initial solution $A_1 :=\tau I_n$ where $I_n \in \R^{n \times n}$ is the identity matrix, and we have
\begin{align*}
    \Phi(A_1) =&~ \sum_{i=1}^m \cosh(\lambda\cdot (\tau - b_i))\\
    \leq &~ e^{\lambda \tau} \sum_{i=1}^m e^{-\lambda b_i}\leq 2^{O(\lambda R)},
\end{align*}
where the last step follows from $|b_i|\leq R$ for all $i\in [m]$.

After $T$ iterations, we have
\begin{align*}
    \Phi(A_{T+1}) \leq & ~ (1-\frac{\epsilon \lambda}{\sqrt{m}})^T \Phi(A_1) + 2  m \\
    \leq & ~ (1-\frac{\epsilon \lambda}{\sqrt{m}})^T \cdot 2^{O(\lambda R) }+ 2  m \\
    \leq & ~ 2^{-\Omega( T  \epsilon \lambda / \sqrt{m} ) + O(\lambda R)} + 2  m
\end{align*}
where the first step follows from applying Lemma~\ref{lem:gradient_descent} for $T$ times, and $\sum_{i=1}^T (1-\epsilon\lambda/\sqrt{m})^{i-1}\epsilon \lambda \sqrt{m}\leq 2m$.

As long as $T= \Omega(R \sqrt{m} / \epsilon )=\Omega(R\sqrt{m}\lambda)$, then we have
\begin{align*}
    \Phi(A_{T+1}) \leq O(m).
\end{align*}

This implies that for any $i\in [m]$,
\begin{align*}
    | u_i^\top A_{T+1} u_i - b_i| \leq &~ \lambda^{-1}\cdot \cosh^{-1}(O(m))\\
    = &~ \lambda^{-1}\cdot O(\log m)\\
    = &~ \delta,
\end{align*}
where we take $R = \Omega(\delta^{-1}\log m)$.  

Therefore, with $T=\widetilde{\Omega}(\sqrt{m}R\delta^{-1})$ iterations, Algorithm~\ref{alg:GD} can achieve that
\begin{align}\label{eq:gd_approximation_guarantee}
    | u_i^\top A_{T+1} u_i - b_i| \leq \delta~~~\forall i\in [m].
\end{align}
The theorem is then proved.
\end{proof}

\section{Spectral Potential function with ground-truth oracle}\label{sec:potential_ground_truth_app}
In this section, we consider the matrix sensing with spectral approximation; that is, we want to obtain a matrix $A$ that is a $\delta$-spectral approximation of the ground-truth matrix $A_\star$, i.e.,
\begin{align*}
    (1-\delta)A_\star\preceq A\preceq (1+\delta)A_\star.
\end{align*}
To do this, instead of performing a series of quadratic measurements, we assume that we have access to an oracle ${\cal O}_{A_\star}$ such that for any matrix $A\in \R^{n\times n}$, the oracle will output a matrix $A_\star^{-1/2}AA_\star^{-1/2}$. Algorithm~\ref{alg:GD_spectral} implements a matrix sensing algorithm with spectral approximation guarantee with the assumption of oracle ${\cal O}_{A_\star}$.

We define the spectral loss function as follows:
\begin{align*}
    \Psi_{\lambda}(A) := \tr[ \cosh ( \lambda ( I - (A_{\star})^{-1/2} A (A_{\star})^{-1/2} )  ) ].
\end{align*}
We will show that $\Psi_\lambda(A)$ can characterize the spectral approximation of $A$ with respect to $A_\star$.

It is easy to see that if we can query an arbitrary $A$ to the ground-truth oracle ${\cal O}_{A_\star}$, then we can definitely recover $A_\star$ exactly by querying ${\cal O}_{A_\star}(I)$. Instead, in Algorithm~\ref{alg:GD_spectral}, we focus on the following process: the initial matrix $A_1$ is given, and in the $t$-th iteration, we first compute
\begin{align*}
    X_t =  \lambda (I - A_{\star}^{-1/2} A_t A_{\star}^{-1/2} )
\end{align*}
and do eigendecompsotion of $X_t$ to obtain $\Lambda_t$ such that $X_t = Q_t \Lambda_t Q_t^{\top}$. Then we update the matrix $A_{t+1}$ by:
\begin{align*}
        A_{t+1} = A_t +  \epsilon \cdot  A_{\star}^{1/2} \sinh(X_t) A_{\star}^{1/2} / \| \sinh(X_t) \|_F.
\end{align*}
We are interested in the number of iterations needed to make $A_t$ be a $\delta$-spectral approximation. We believe this example will provide some insight into this problem, and we leave the question of spectral-approximated matrix sensing without the ground-truth oracle to future work.

\begin{algorithm*}\caption{Matrix Sensing with Spectral Approximation.}\label{alg:GD_spectral}
\begin{algorithmic}[1]
\Procedure{GradientDescent}{${\cal O}_{A_\star}$, $A_1$} 
    \For{$t = 1 \to T$}
        \State $X_t\gets \lambda\cdot (I_n - {\cal O}_{A_\star}(A_t))$
        \State $Q_t \Lambda_t Q_t^\top\gets$ Eigendecomposition of $X_t$\Comment{It takes $O(n^\omega)$-time}
        \State $Y_t \gets Q_t \cdot \sinh(\Lambda_t)\cdot Q_t^\top$\Comment{$Y_t=\sinh(X_t)$. It takes $O(n^2)$-time}
        \State $A_{t+1} \gets A_t + \epsilon  \cdot {\cal O}_{A_\star}(Y_t) / \| Y_t \|_F$\Comment{It takes $O(n^2)$-time}
    \EndFor
    \State \Return $A_{T+1}$
\EndProcedure
\end{algorithmic}
\end{algorithm*}

\begin{lemma}[Progress on the spectral potential function]\label{lem:potential_func_loss}
Let $c \in (0,1)$ denote a sufficiently small positive constant. 
We define $X_t$ as follows:
\begin{align*}
    X_t := \lambda (I - (A_{\star})^{-1/2} A_t (A_{\star})^{-1/2} )
\end{align*}

Let 
\begin{align*}
    A_{t+1} = A_t
    +  \epsilon \cdot \lambda (A_{\star})^{1/2} \sinh(X_t) (A_{\star})^{1/2} / \| \lambda \cdot \sinh(X_t) \|_F.
\end{align*}

For any $\epsilon\in (0, 1)$ and $\lambda \geq 1$ such $\lambda \epsilon \leq c$, 
we have for any $t>0$,

\begin{align*}
\Psi_{\lambda}(A_{t+1})   \leq (1 - 0.9\epsilon \lambda /\sqrt{n} ) \Psi_{\lambda}(A_t) + \epsilon \lambda \sqrt{n}.
\end{align*}
\end{lemma}

\begin{proof}

We can compute
\begin{align}\label{eq:derivative}
    & ~ \Psi_{\lambda}(A_{t+1}) - \Psi_{\lambda}(A_t) \notag \\
    = & ~ \tr[ \cosh ( X_{t+1} )] - \tr[ \cosh ( X_t  ) ] \notag\\
    \leq & ~ - \lambda \cdot \tr[ \sinh( X_t ) \cdot ( (A_{\star})^{-1/2} (A_{t+1} - A_t )  (A_{\star})^{-1/2} ) ] \notag\\
    + & ~ O(1) \cdot \lambda^2 \cdot \tr[\cosh( X_t ) \cdot ( (A_{\star})^{-1/2} (A_t - A_{t+1} )  (A_{\star})^{-1/2} )^2 ] \notag\\
    = & ~ - \Delta_1 + O(1) \cdot \Delta_2, 
\end{align}
the first step is by expanding by definition, the second step is by Taylor expanding the first term at the point $I-(A_{\star})^{-1/2}A_t(A_{\star})^{-1/2}$ (via Lemma~\ref{lem:jn08}), and the last step is by definition of $\Delta_1$ and $\Delta_2$.

To further simplify proofs, we define
\begin{align*}
    \nabla \Psi_{\lambda}(A_t) := & ~ \lambda \cdot (A_{\star})^{1/2} \sinh( X_t ) (A_{\star})^{1/2} \\
    \wt{\nabla} \Psi_{\lambda}(A_t) := & ~ \lambda \cdot \sinh( X_t )  \\
    \wt{\Delta} \Psi_{\lambda}(A_t) := & ~ \lambda \cdot \cosh( X_t )  
\end{align*}

To maximize the gradient progress, we should choose
\begin{align*}
    A_{t+1} = A_t + \epsilon \cdot \nabla \Psi_{\lambda}(A_t) / \| \wt{\nabla} \Psi_{\lambda}(A_t) \|_F
\end{align*}

Then 
\begin{align}\label{eq:derivative_first_moment_A_t_F}
    \Delta_1 = & ~ (\epsilon \lambda^2) \cdot \tr[  \sinh^2( X_t ) ] / \| \wt{\nabla} \Psi_{\lambda}(A_t) \|_F \notag \\
    = & ~ \epsilon \cdot \| \wt{\nabla} \Psi_{\lambda}(A_t) \|_F^2 /  \| \wt{\nabla} \Psi_{\lambda}(A_t) \|_F \notag \\
    = & ~ \epsilon \cdot \| \wt{\nabla} \Psi_{\lambda}(A_t) \|_F 
\end{align}
and
\begin{align}\label{eq:derivative_second_moment_A_t_F}
    \Delta_2 = & ~ \epsilon ^2  \lambda^4 \cdot \tr[ \cosh( X_t ) \cdot \sinh^2( \lambda ( X_t )  ] / \| \wt{\nabla} \Psi_{\lambda}(A_t) \|_F^2 \notag \\
    = & ~ \epsilon^2 \lambda \cdot \tr[ \wt{\Delta} \Psi_{\lambda}(A_t) \cdot \wt{\nabla} \Psi_{\lambda}(A_t)^2 ] / \| \wt{\nabla} \Psi_{\lambda}(A_t) \|_F^2 \notag\\
    \leq & ~ \epsilon^2 \lambda \cdot \| \wt{\Delta} \Psi_{\lambda}(A_t) \|_F \cdot \| \wt{\nabla} \Psi_{\lambda}(A_t)^2 \|_F / \| \wt{\nabla} \Psi_{\lambda}(A_t) \|_F^2 \notag\\
    \leq & ~ \epsilon^2 \lambda \cdot \| \wt{\Delta} \Psi_{\lambda}(A_t) \|_F \notag\\
    \leq & ~ \epsilon^2 \lambda \cdot ( \lambda \sqrt{n} + \| \wt{\nabla} \Psi_{\lambda}(A_t) \|_F  )
\end{align}
where { the first step follows from the definition of $\Delta_2 $, the second step comes from the definition of $\wt{\Delta} \Psi_{\lambda}(A_t) $ and $\wt{\nabla} \Psi_{\lambda}(A_t)$,} { the third step follows that $\|A B\|_F \leq \|A\|_F \|B\|_F$,} the forth step follows from $\| x \|_4^2 \leq \| x \|_2^2$, and the fifth step follows from Part 1 of Lemma~\ref{lem:property_sinh_cosh_matrix}.

Now, we need to lower bound $\| \wt{\nabla} \Psi_{\lambda}(A_t) \|_F $, we have
\begin{align}\label{eq:derivative_phi_A_t_F}
    \| \wt{\nabla} \Psi_{\lambda}(A_t) \|_F = & ~ ( \tr[ \lambda^2 \sinh^2(X_t) ] )^{1/2} \notag \\
    \geq & ~ \frac{\lambda}{\sqrt{n}} ( \tr[  \cosh(X_t) ] -  n ) \notag \\
    = & ~ \frac{\lambda}{\sqrt{n}} ( \Psi_{\lambda}(A_t) -  n ) 
\end{align}
where the second step follows from Part 2 in Lemma~\ref{lem:property_sinh_cosh_matrix}.

We know that

Then, we have
\begin{align*}
    & ~ \Psi_{\lambda}(A_{t+1}) - \Psi_{\lambda}(A_t) \\
    \leq & ~ - \epsilon \| \wt{\nabla} \Psi_{\lambda}(A_t) \|_F + \epsilon^2 \lambda ( \sqrt{n} + \| \wt{\nabla} \Psi_{\lambda}(A_t) \|_F  ) \\
    \leq & ~ - 0.9 \epsilon  \| \wt{\nabla} \Psi_{\lambda}(A_t) \|_F + \epsilon^2 \lambda^2 \sqrt{n} \\
    \leq & ~ - 0.9 \epsilon \lambda \frac{1}{\sqrt{n}}  \Psi_{\lambda}(A_t)  + \epsilon \lambda \sqrt{n}
\end{align*}
where the first step {  follows from  Eq.~\eqref{eq:derivative_first_moment_A_t_F} and Eq.~\eqref{eq:derivative_second_moment_A_t_F} , the second steps comes from $\epsilon \in (0, 0.01)$, the third step comes from Eq.~\eqref{eq:derivative_phi_A_t_F} and $\epsilon \lambda \leq 1$. }

Finally, we complete the proof. 

\end{proof}

\begin{lemma}[Small spectral potential implies good spectral approximation]
Let $A\in \R^{n\times n}$ be symmetric, and $\lambda >0$. Suppose $\Psi_\lambda(A)\leq p$ for some $p > 1$. Then, we have
\begin{align*}
    (1-\delta) A_\star \preceq A \preceq (1+\delta)A_\star
\end{align*}
for $\delta = O(\lambda^{-1}\log p)$.
\end{lemma}

\begin{proof}
By the definition of $\Psi_\lambda(A)$,  $\Psi_\lambda(A)\leq p$ implies that for any $i\in [n]$,
\begin{align*}
    \cosh(\lambda(1-\lambda_i(A_\star^{-1/2}AA_\star^{-1/2})))\leq p,
\end{align*}
or equivalently,
\begin{align*}
    \left|(1-\lambda_i(A_\star^{-1/2}AA_\star^{-1/2}))\right|\leq O(\lambda^{-1}\log p).
\end{align*}
Hence, we have
\begin{align*}
    (1-\delta)I_n \preceq A_\star^{-1/2}AA_\star^{-1/2} \preceq (1+\delta)I_n,
\end{align*}
where $\delta := O(\lambda^{-1}\log p)$. Therefore, by multiplying $A_\star^{-1/2}$ on both sides, we get that
\begin{align*}
    (1-\delta) A_\star \preceq A \preceq (1+\delta)A_\star,
\end{align*}
which completes the proof of the lemma.
\end{proof}
\section{Gradient descent with General Measurements}\label{sec:gd_non_orthogonal}
  
In this section, we analyze the potential decay by gradient descent with non-orthogonal measurements.  The main result of this section is Lemma~\ref{lem:gradient_descent_rho} in below.

We first recall the definition of the potential function $\Phi_{\lambda}(A)$:
\begin{align*}
    \Phi_{\lambda}(A) := \sum_{i=1}^m \cosh ( \lambda ( u_i^\top A u_i - b_i ) ),
\end{align*}
its gradient $\nabla \Phi_{\lambda}(A)\in \R^{n\times n}$:
\begin{align}\label{eq:gradient_phi_A_rho}
    \nabla \Phi_{\lambda}(A) = \sum_{i=1}^m u_i u_i^\top \lambda \sinh\left( \lambda (u_i^\top A u_i - b_i)\right), 
\end{align}
and its Hessian $\nabla^2 \Phi_{\lambda}(A)\in \R^{n^2\times n^2}$:
\begin{align*}
    \nabla^2 \Phi_{\lambda}(A) = \sum_{i=1}^m ( u_i u_i^\top ) \otimes ( u_i u_i^\top ) \lambda^2 \cosh( \lambda (u_i^\top A u_i - b_i) ). 
\end{align*}

\begin{lemma}[Progress on entry-wise potential with general measurements]\label{lem:gradient_descent_rho}
Assume that $|u_i^{\top}  u_j| \leq \rho$ and $\rho \leq \frac{1}{10m}$, for any $i,j \in [m]$ and $\| u_i\|^2 = 1$. Let $c \in (0,1)$ denote a sufficiently small positive constant. Then, for any $\epsilon,\lambda>0$ such that $\epsilon\lambda \leq c$, 
we have for any $t>0$,
\begin{align*}
    \Phi_{\lambda} ( A_{t+1} ) \leq (1-0.9 \frac{ \lambda \epsilon }{\sqrt{m} }) \cdot \Phi_{\lambda} (A_t) +  \lambda \epsilon \sqrt{m}
\end{align*}

\end{lemma}

\begin{proof}
We first have
\begin{align}\label{eq:gd_define_Delta_1_and_Delta_2_rho}
    & ~ \Phi_{\lambda}(A_{t+1})- \Phi_{\lambda} (A_t) \notag \\
    \leq & ~ \langle \nabla \Phi_{\lambda} (A_t) , (A_{t+1} - A_t) \rangle 
    + O(1) \langle \nabla^2 \Phi_{\lambda}(A_t), (A_{t+1} - A_t) \otimes (A_{t+1} - A_t) \rangle \notag \\
    := & ~ - \Delta_1 + O(1) \cdot \Delta_2,
\end{align}
which follows from Corollary~\ref{cor:matrix_der_trace}.

We choose
\begin{align}\label{eq:A_t1_update_rho}
    A_{t+1} = A_t - \epsilon \cdot \nabla \Phi_{\lambda}(A_t) / \| \nabla \Phi_{\lambda}(A_t) \|_F.
\end{align}

We can bound
\begin{align}\label{eq:tr_phi_At_first_moment_rho} 
 \Delta_1 = & ~ -\tr[\nabla \Phi_{\lambda}(A_t) (A_{t+1} - A_t)] \notag \\
=  & ~ \epsilon \cdot \|\nabla \Phi_{\lambda} (A_t) \|_F.
\end{align}

For $\| \Phi_{\lambda} (A_t) \|_F^2$,
\begin{align}\label{eq:phi_At_F_norm_rho}
& ~ \frac{1}{\lambda^2} \| \nabla \Phi_{\lambda} (A_t) \|_F^2 \notag \\
= & ~   \tr[ (\sum_{i=1}^m u_i u_i^\top \sinh( \lambda( u_i^\top A_t u_i - b_i ) )  )^2 ] \notag\\
= & ~   \tr[ \sum_{i=1}^m  \sinh^2 ( \lambda( u_i^\top A_t u_i - b_i ) )  ] \notag \\
 + &~ \tr[ \sum_{i =1}^{m}\sum_{j \neq i}^{m} (u_i u_i^\top)(u_j u_j^{\top}) \sinh ( \lambda( u_i^\top A_t u_i - b_i ) ) \cdot \sinh ( \lambda( u_j^\top A_t u_j - b_j ) )]   \notag \\
 \geq &~0.9 \tr[ \sum_{i=1}^m  \sinh^2 ( \lambda( u_i^\top A_t u_i - b_i ) )] \notag \\
\geq & ~ 0.9\frac{1}{m} ( \sum_{i=1}^m \cosh ( \lambda ( u_i^\top A_t u_i - b_i ) ) - m )^2 \notag \\
= & ~ 0.9\frac{1}{m} ( \Phi_{\lambda} (A_t) - m )^2,
\end{align}
where the first step follows from Eq.~\eqref{eq:gradient_phi_A_rho},
the second steps follow from partitioning based on whether $i = j$ and $\| u_i \|_2 = 1$, the third step comes from Claim~\ref{claim:sum_ij_off_diagonal}, the fourth step in Eq.~\eqref{eq:phi_At_F_norm_rho} follows from Part 2 in Lemma~\ref{lem:property_sinh_cosh_scalar},  the fifth step follows from the definition of $\Phi_{\lambda}(A)$.

Thus,
\begin{align}\label{eq:bound_gd_Delta_1_rho}
   \Delta_1 = & ~ -\tr[\nabla \Phi_{\lambda}(A_t) (A_{t+1} - A_t)] \notag \\
   \geq & ~ \lambda \epsilon \cdot \frac{1}{ \sqrt{m} } ( \Phi_{\lambda}(A_t) - m )  .
\end{align}

For simplicity, we define
\begin{align*}
    z_{t,i} := \lambda (u_i^\top A_t u_i - b_i).
\end{align*}

We need to compute this $\Delta_2$. For simplificity, we consider $\Delta_2 \cdot (\frac{1}{\epsilon \lambda})^2 \cdot \| \nabla \Phi_{\lambda} (A_t) \|_F^2$, which can be expressed as:
\begin{align}\label{eq:tr_phi_At_second_moment_rho}
&\Delta_2 \cdot (\frac{1}{\epsilon \lambda})^2 \cdot \| \nabla \Phi_{\lambda} (A_t) \|_F^2\notag\\
     = & ~ \frac{1}{(\lambda \epsilon)^2} \tr[ \nabla^2 \Phi_{\lambda}(A_t) \cdot (A_{t+1} - A_t) \otimes (A_{t+1} - A_t) ] \cdot 
     \| \nabla \Phi_{\lambda}(A_t) \|_F^2 \notag\\
    = & ~  \tr\Big[ \nabla^2 \Phi_{\lambda}(A_t) \cdot 
    ( \sum_{i=1}^m u_i u_i^\top  \sinh( z_{t,i} ) ) \otimes ( \sum_{i=1}^m u_i u_i^\top \sinh( z_{t,i} ) )\Big] \notag\\
    = & ~ \tr[ \nabla^2 \Phi_{\lambda}(A_t)  ( \sum_{i,j}  \sinh( z_{t,i} )\sinh( z_{t,i} ) (u_i u_i^\top \otimes u_ju_j^\top) ) ] \notag\\
    = & ~ \tr[ \nabla^2 \Phi_{\lambda}(A_t)  ( \sum_{i=1}^m \sinh^2( z_{t,i} )) (u_i u_i^\top \otimes u_iu_i^\top) ) ] \notag\\
    + & ~  \tr[ \nabla^2 \Phi_{\lambda}(A_t)  ( \sum_{i\neq j} \sinh( z_{t,i} )\sinh( z_{t,j} ) (u_i u_i^\top \otimes u_j u_j^\top) ) ] \notag\\
    = & ~  Q_1 + Q_2  ,
\end{align}
where 
\begin{align}\label{eq:def_Q1_rho}
Q_1:= \tr\Big[ \nabla^2 \Phi_{\lambda}(A_t) \cdot ( \sum_{i=1}^m \sinh^2( z_{t,i} )) (u_i u_i^\top \otimes u_iu_i^\top) ) \Big]
\end{align}
denotes the diagonal term, and
\begin{align}\label{eq:def_Q2_rho}
    Q_2:=\tr\Big[ \nabla^2 \Phi_{\lambda}(A_t) \cdot ( \sum_{i\neq j} \sinh( z_{t,i} )\sinh( z_{t,j} ) (u_i u_i^\top \otimes u_j u_j^\top) ) \Big]
\end{align}
denotes the off-diagonal term. The first step comes from the definition of $\Delta_2$,
the second step follows from replacing $A_{t+1} - A_t$ using Eq.~\eqref{eq:A_t1_update_rho}, the third step follows that we extract the scalar values from Kronecker product, the fourth step comes from splitting into two partitions based on whether $i = j$, the fifth step comes from the definition of $Q_1$ and $Q_2$.

Thus,
\begin{align}\label{eq:bound_gd_Delta_2_rho}
    \Delta_2 \leq & ~ (\epsilon \lambda)^2 (Q_1 + Q_2) / \| \nabla \Phi_{\lambda}(A_t) \|_F^2 \notag \\
    = & ~ 1.3(\epsilon \lambda )^2 \cdot ( \sqrt{m} + \frac{1}{\lambda} \| \nabla \Phi_{\lambda}(A_t) \|_F ).
\end{align}
where the second step follows from  Claim~\ref{cla:gd_Q1_rho} and Claim~\ref{cla:gd_Q2_rho}.

Hence, we have
\begin{align*}
    & ~ \Phi_{\lambda} (A_{t+1}) - \Phi_{\lambda} (A_t) \\
    \leq & ~ - \Delta_1    + O(1) \cdot \Delta_2 \\
    \leq & ~ - \epsilon \| \nabla \Phi_{\lambda}(A_t) \|_F   + O(1) (\epsilon \lambda)^2 (\sqrt{m} + \frac{1}{\lambda} \| \nabla \Phi_{\lambda}(A_t) \|_F ) \\
    \leq & ~ - 0.9 \epsilon \| \Phi_{\lambda}(A_t) \|_F+ O(\epsilon \lambda)^2 \sqrt{m} \\
    \leq & ~ -0.9 \epsilon \lambda \frac{1}{\sqrt{m}} ( \Phi_{\lambda}(A_t) - m ) +O(\epsilon \lambda)^2 \sqrt{m} \\
    \leq & ~  -0.9 \epsilon \lambda \frac{1}{\sqrt{m}} \Phi_{\lambda}(A_t) + \epsilon \lambda \sqrt{m},
\end{align*}
where the first step follows from Eq.~\eqref{eq:gd_define_Delta_1_and_Delta_2_rho}, the second step follows from Eq.~\eqref{eq:bound_gd_Delta_1_rho} and Eq.~\eqref{eq:bound_gd_Delta_2_rho},  the third step follows from $\epsilon \lambda \in (0, 0.01)$, the fourth step follows from Lemma~\ref{lem:cosh_sinh}, and the final step follows that extracting the constant term from the summation.

The lemma is then proved.
\end{proof}

We prove some technical claims in below.

\begin{claim}\label{claim:sum_ij_off_diagonal}
It holds that:
\begin{align*}
    & \sum_{i\ne j\in [m]} \langle u_i,u_j\rangle^2 \sinh ( \lambda( u_i^\top A_t u_i - b_i ) ) \sinh ( \lambda( u_j^\top A_t u_j - b_j ) ) \\
    \leq &~ 0.1 \sum_{i=1}^m  \sinh^2 ( \lambda( u_i^\top A_t u_i - b_i ) ) 
\end{align*}
\end{claim}
\begin{proof}
We define $R_{i,j}$ and $R$ as follows:
\begin{align*}
    R_{i,j} = &~\sinh ( \lambda( u_i^\top A_t u_i - b_i ) )  \sinh ( \lambda( u_j^\top A_t u_j - b_j ) ) \\
     R = &~\tr[ \sum_{i =1}^{m}\sum_{j \neq i}^{m} (u_i u_i^\top)(u_j u_j^{\top}) \sinh ( \lambda( u_i^\top A_t u_i - b_i ) ) \cdot 
     \sinh ( \lambda( u_j^\top A_t u_j - b_j ) )]
\end{align*}

Then we can upper bound $|R|$ by:
\begin{align*}
    |R| =  &~ \tr[\sum_{i =1}^{m}\sum_{j \neq i}^{m} |(u_i u_i^\top)(u_j u_j^{\top})| |R_{i,j}|] \\
    \leq &~ \rho^2 \tr[\sum_{i =1}^{m}\sum_{j \neq i}^{m}  |R_{i,j}|] \\
    \leq &~  \frac{\rho^2 }{2}\tr[\sum_{i =1}^{m}\sum_{j \neq i}^{m} (R_{i,i} + R_{j,j})] \\
    \leq &~ m \rho^2 \tr[\sum_{i=1}^{m} R_{i,i} ] \\
    \leq &~ 0.1\tr[\sum_{i=1}^m R_{i,i}]
\end{align*}
where the first step follows $|ab| = |a||b|$, the second step follows $|u_i^{\top}  u_j| \leq \rho$, the third step follows that $|ab| \leq \frac{a^2 + b^2}{2}$, the fourth step follows from the summation over $j$, and the fifth step comes from $m\rho^2 \leq 0.1$.
\end{proof}

\begin{claim}\label{cla:gd_Q1_rho}
For $Q_1$ defined in Eq.~\eqref{eq:def_Q1_rho}, we have
\begin{align*}
    Q_1 \leq 1.1\Big( \sqrt{m} + \frac{1}{\lambda} \| \nabla \Phi_{\lambda}(A_t) \|_F \Big) \cdot  \| \nabla \Phi_{\lambda} (A_t) \|_F^2.
\end{align*}
\end{claim}
\begin{proof}

For simplicity, we define $z_{t,i}$ to be
\begin{align*}
    z_{t,i} := \lambda ( u_i^\top A_t u_i - b_i ) .
\end{align*}
Recall that
\begin{align*}
    \nabla^2 \Phi_{\lambda}(A_t) = \lambda^2 \cdot \sum_{i=1}^m ( u_i u_i^\top ) \otimes ( u_i u_i^\top ) \cosh( z_{t,i} ) .
\end{align*}

For $Q_1$, we have
\begin{align}\label{eq:Q1_rho}
   Q_1 = & ~  \tr[ \nabla^2 \Phi_{\lambda}(A_t) \sum_{i=1}^m  \sinh^2( z_{t,i} ) (u_i u_i^\top \otimes u_i u_i^\top) ) ] \notag\\
   = & ~ \lambda^2 \cdot \tr[ \sum_{i=1}^m  \cosh( z_{t,i} ) ( u_i u_i^\top ) \otimes ( u_i u_i^\top )  
   \cdot   \sum_{i=1}^m  \sinh^2( z_{t,i} ) (u_i u_i^\top ) \otimes ( u_i u_i^\top)   ] \notag\\
   = & ~ \lambda^2 \cdot \sum_{i=1}^m \tr[ \cosh( z_{t,i} ) \sinh^2( z_{t,i} )\cdot 
    ( u_i u_i^\top  u_i u_i^\top )  \otimes ( u_i u_i^\top u_i u_i^\top ) ] \notag\\
+ &~  \lambda^2 \cdot \sum_{i=1}^m \sum_{j \neq i}^{m} \tr[ \cosh( z_{t,i} )\sinh^2( z_{t,j} ) \cdot 
 (u_i u_i^\top u_j u_j^{\top} ) \otimes (u_j u_j^{\top} u_i u_i^\top)] \notag \\
   = & ~ \lambda^2 \cdot \sum_{i=1}^m  \cosh( z_{t,i} ) \sinh^2( z_{t,i} ) 
    \notag \\ 
   + & ~  \lambda^2 \cdot \sum_{i=1}^m \sum_{j \neq i}^{m} \tr[ \cosh( z_{t,i} )\sinh^2( z_{t,j} ) \cdot 
 (u_i u_i^\top u_j u_j^{\top} ) \otimes (u_j u_j^{\top} u_i u_i^\top)] \notag \\
   \leq & ~  1.1 \lambda ^2 \cdot (  \sum_{i=1}^m  \cosh^2( z_{t,i} ) )^{1/2}
   \cdot  ( \sum_{i=1}^m \sinh^4( z_{t,i} ) )^{1/2} \notag \\
   \leq & ~ 1.1 \lambda^2 \cdot B_1 \cdot B_2,
\end{align}
where {  the first step comes from the definition of $Q_1$, the second step comes from the definition of $\nabla^2 \Phi_{\lambda}(A_t)$,}
the third step follows from $(A \otimes B) \cdot (C \otimes D) = (AC) \otimes (BD)$ and partition the terms based on whether $i = j$, the fourth step comes from $\|u_i\| = 1$ and $\tr[ (u_i  u_i^\top) \otimes (u_i  u_i^\top) ] = 1$, and the fifth step comes from Cauchy–Schwarz inequality and Claim~\ref{claim:q1_helper_off_diagonal}.

\begin{claim}\label{claim:q1_helper_off_diagonal}
We can bound the off-diagonal entries by: 
\begin{align*}
  &~   |\lambda^2 \cdot \sum_{i=1}^m \sum_{j \neq i}^{m} \tr[ \cosh( z_{t,i} )\sinh^2( z_{t,j} ) \cdot 
  (u_i u_i^\top u_j u_j^{\top} ) \otimes (u_j u_j^{\top} u_i u_i^\top)]| \notag \\
 \leq &~ 0.1 \lambda^2(\sum_{i=1}^m(\cosh(z_{t,1}))^{1/2} \cdot (\sum_{i=1}^m \sinh^4( z_{t,i}))^{1/2}   \\  
\end{align*}
\end{claim}
\begin{proof}

\begin{align*}
      &~   |\lambda^2 \cdot \sum_{i=1}^m \sum_{j \neq i}^{m} \tr[ \cosh( z_{t,i} )\sinh^2( z_{t,j} ) \cdot 
       (u_i u_i^\top u_j u_j^{\top} ) \otimes (u_j u_j^{\top} u_i u_i^\top)]| \notag \\
      \leq &~ \rho^2 \lambda^2 |\sum_{i=1}^m \sum_{j \neq i}^{m} \cosh( z_{t,i} )\sinh^2( z_{t,j} )| \\
      \leq &~  \rho^2 \lambda^2 (\sum_{i=1}^m \sum_{j \neq i}^{m} (\cosh^2( z_{t,i} ))^{1/2} \cdot (\sum_{i=1}^m \sum_{j \neq i}^{m}\sinh^4( z_{t,j}))^{1/2} \\
      \leq &~ m\rho^2 \lambda^2  (\sum_{i=1}^m(\cosh^2(z_{t,i}))^{1/2} \cdot (\sum_{i=1}^m \sinh^4( z_{t,i}))^{1/2} \\
      \leq &~ 0.1 \lambda^2(\sum_{i=1}^m(\cosh^2(z_{t,i}))^{1/2} \cdot (\sum_{i=1}^m \sinh^4( z_{t,i}))^{1/2} 
\end{align*}
where the first step comes from $|\langle u_i, u_j\rangle | \leq \rho$, the second step comes from Cauchy–Schwarz inequality, the third step follows from summation over $m$ terms, and the fourth step comes from $\rho^2 m \leq 0.1$.
\end{proof}

For the term $B_1$, we have
\begin{align}\label{eq:b1_rho}
    B_1 = & ~ (  \sum_{i=1}^m \cosh^2( \lambda (u_i^\top A_t u_i - b_i ) ) )^{1/2} \notag \\
    \leq & ~ \sqrt{m} + \frac{1}{\lambda} \| \nabla \Phi_{\lambda}(A_t) \|_F,
\end{align}
where the second step follows Part 1 of Lemma~\ref{lem:property_sinh_cosh_scalar}.

For the term $B_2$, we have
\begin{align}\label{eq:b2_rho}
    B_2 = & ~( \sum_{i=1}^m \sinh^4( \lambda( u_i^\top A_t u_i - b_i ) ) )^{1/2} \notag  \\
    \leq & ~ \frac{1}{\lambda^2} \| \nabla \Phi_{\lambda} (A_t) \|_F^2,
\end{align}
where the second step follows from $\| x \|_4^2 \leq \| x \|_2^2$. This implies that
\begin{align*}
    Q_1 \leq & ~ 1.1\lambda^2 \cdot B_1 \cdot B_2 \\
    \leq & ~ 1.1\lambda^2 \cdot ( \sqrt{m} + \frac{1}{\lambda} \| \nabla \Phi_{\lambda}(A_t) \|_F ) \cdot \frac{1}{\lambda^2} \| \nabla \Phi_{\lambda} (A_t) \|_F^2 \\
    = & ~ 1.1( \sqrt{m} + \frac{1}{\lambda} \| \nabla \Phi_{\lambda}(A_t) \|_F ) \cdot  \| \nabla \Phi_{\lambda} (A_t) \|_F^2 .
\end{align*}
This completes the proof.
\end{proof}

\begin{claim}\label{cla:gd_Q2_rho} 
For $Q_2$ defined in Eq.~\eqref{eq:def_Q2_rho}, we have:
\begin{align*}
    Q_2 \leq 0.2 \lambda^2 ( \sqrt{m} + \frac{1}{\lambda} \| \nabla \Phi_{\lambda}(A_t) \|_F ) \cdot  \| \nabla \Phi_{\lambda} (A_t) \|_F^2 
\end{align*}
\end{claim}
\begin{proof}
Because in $Q_2$ we have :
\begin{align}\label{eq:u_ell_i_j_product_rho}
 Q_2 =  & ~ \lambda^2 \tr[\sum_{\ell=1}^{m}(\cosh(z_{t,\ell}) \cdot u_{\ell} u_{\ell}^{\top} \otimes u_{\ell} u_{\ell}^{\top}) \cdot 
 \sum_{i \neq j}^{m}(\sinh( z_{t,i} )\sinh( z_{t,j} ) \cdot u_i u_i^{\top} \otimes u_j u_j^{\top})] \notag \\
    = & ~ \lambda^2\tr[\sum_{\ell=1}^{m} \sum_{i \neq j}^{m} \cosh(z_{t,\ell} )\sinh( z_{t,i} )\sinh( z_{t,j} ) \cdot 
     (u_{\ell} u_{\ell}^{\top} u_i u_i^{\top}) \otimes (u_{\ell} u_{\ell}^{\top} u_j u_j^{\top})] \notag \\
    \leq & ~  \lambda^2\rho^2 \sum_{\ell=1}^{m} \sum_{i \neq j}^{m} \cosh(z_{t,\ell} )(\sinh^2( z_{t,i} ) + \sinh^2( z_{t,j} )) \notag \\
    \leq &~ 2m \lambda^2\rho^2 \sum_{\ell=1}^{m} \sum_{i = 1}^{m} \cosh(z_{t,\ell} )\sinh^2( z_{t,i} ) \notag \\
    \leq &~ 2m^2 \lambda^2\rho^2 \sum_{i = 1}^{m} \cosh(z_{t,i} )\sinh^2( z_{t,i} ) \notag \\
    \leq &~ 2m^2 \lambda^2\rho^2 (\sum_{i=1}^m(\cosh^2(z_{t,i}))^{1/2}  (\sum_{i=1}^m \sinh^4( z_{t,i}))^{1/2} \notag \\
    \leq &~ 0.2 \lambda^2 ( \sqrt{m} + \frac{1}{\lambda} \| \nabla \Phi_{\lambda}(A_t) \|_F ) \cdot  \| \nabla \Phi_{\lambda} (A_t) \|_F^2 
\end{align}
where the second step follows from $(A \otimes B) \cdot (C \otimes D) = (AC) \otimes (BD)$, the third step follows Cauchy–Schwarz inequality and $|\langle u_i, u_j\rangle | \leq \rho$, the fourth step follows from combining $\sinh^2( z_{t,i} )$ and $\sinh^2( z_{t,j} )$, the fifth step comes from summation over $m$ terms, and the sixth step comes from Cauchy–Schwarz inequality and the seventh step follows from Eq.~\eqref{eq:b1_rho} and Eq.~\eqref{eq:b2_rho} and $m^2 \rho^2 \leq 0.1$.

\end{proof}

\section{Stochastic Gradient Descent for General Measurements}\label{sec:sgd_general}

In this section, we further extend the general measurement  where $\{u_i\}_{i\in [m]}$ are non-orthogonal vectors and $|u_i^{\top}  u_j| \leq \rho$ to the convergence analysis of the stochastic gradient descent matrix sensing algorithm. Algorithm~\ref{alg:stochastic_gradient_descent_general} implements the stochastic gradient descent version of the matrix sensing algorithm.

In Algorithm~\ref{alg:stochastic_gradient_descent_general}, at each iteration $t$, we first compute the stochastic gradient descent by:
\begin{align*}
    \nabla \Phi_{\lambda} (A_t, {\cal B}_t) \gets \frac{m}{B} \sum_{i \in {\cal B}_t} u_i u_i^\top \lambda \sinh( \lambda z_{i} )
\end{align*}
then we update the matrix with the gradient:
\begin{align*}
    A_{t+1} \gets A_t - \epsilon \cdot \nabla \Phi_{\lambda}(A_t,{\cal B} _t) / \| \nabla \Phi_{\lambda}(A_t) \|_F
\end{align*}
At the end of each iteration, we update $z_i$ by:
\begin{align*}
    z_i\gets z_i - \epsilon  \lambda m w_{i,j}^2\sinh(\lambda z_{j}) / (\| \nabla \Phi_{\lambda}(A_t) \|_F B \;\; \forall i \in [m], j \in {\cal B}_t
\end{align*}
We are interested in studying the time complexity and convergence analysis under the general measurement assumption.

\begin{lemma}[Cost-per-iteration of stochastic gradient descent for general measurements]\label{lem:gd_cost_per_iter_general} Algorithm~\ref{alg:stochastic_gradient_descent_general} takes $O(mn^2)$-time for preprocessing and 
each iteration takes
$
    O(Bn^2+m^2)
$-time.
\end{lemma}
\begin{proof}
Since $u_i$'s are no longer orthogonal, we need to compute $\|\nabla \Phi_\lambda (A_t)\|_F$ in the following way:
\begin{align*}
    &\|\nabla \Phi_\lambda (A_t)\|_F^2\\
    =&~ \tr\Big[\Big(\sum_{i=1}^m u_i u_i^\top \lambda \sinh( \lambda (\lambda z_{t,i}))\Big)^2\Big]\\
    = &~ \lambda^2\sum_{i,j=1}^m \langle u_i, u_j\rangle^2 \sinh( \lambda (\lambda z_{t,i}))\sinh( \lambda (\lambda z_{t,j}))\\
    = &~ \lambda^2\sum_{i,j=1}^m w_{i,j}^2 \sinh( \lambda (\lambda z_{t,i}))\sinh( \lambda (\lambda z_{t,j})).
\end{align*}
Hence, with $\{z_{t,i}\}_{i\in [m]}$, we can compute $\|\nabla \Phi_\lambda (A_t)\|_F$ in $O(m^2)$-time.

Another difference from the orthogonal measurement case is the update for $z_{t+1,i}$. Now, we have
\begin{align*}
        &z_{t+1,i}-z_{t,i}\\
    = &~ u_i^\top (A_{t+1}-A_t)u_i\\
    = &~  -\frac{\epsilon}{\| \nabla \Phi_{\lambda}(A_t) \|_F}\cdot u_i^\top \nabla \Phi_{\lambda}(A_t,{\cal B} _t)u_i\\
    = &~ -\frac{\epsilon \lambda m }{\| \nabla \Phi_{\lambda}(A_t) \|_F B} \sum_{j\in {\cal B}_t}u_i^\top u_j u_j^\top u_i \cdot  \sinh(\lambda z_{t,j})\\
    = &~ -\frac{\epsilon \lambda m }{\| \nabla \Phi_{\lambda}(A_t) \|_F B} \sum_{j\in {\cal B}_t}w_{i,j}^2 \cdot  \sinh(\lambda z_{t,j}).
\end{align*}
Hence, each $z_{t+1,i}$ can be computed in $O(B)$-time. And it takes $O(mB)$-time to update all $z_{t+1,i}$.

The other steps' time costs  are quite clear from Algorithm~\ref{alg:stochastic_gradient_descent_general}.
\end{proof}

\begin{lemma}[Progress on expected potential with general measurements]\label{lem:sgd_potential_general}
Assume that $|u_i^{\top}  u_j| \leq \rho$ and $\rho \leq \frac{1}{10m}$, for any $i,j \in [m]$ and $\| u_i\|^2 = 1$. Let $c \in (0,1)$ denote a sufficiently small positive constant. Then, for any $\epsilon,\lambda>0$ such that $\epsilon \lambda \leq c \frac{|{\cal B}_t|}{m}$, 
we have for any $t>0$,
\begin{align*}
    \E[\Phi_{\lambda} ( A_{t+1} )] \leq (1-0.9 \frac{ \lambda \epsilon }{\sqrt{m} }) \cdot \Phi_{\lambda} (A_t) +  \lambda \epsilon \sqrt{m}
\end{align*}

\end{lemma}
The proof is a direct generalization of Lemma~\ref{lem:stochastic_gradient_descent} and is very similar to Lemma~\ref{lem:gradient_descent_rho}. Thus, we omit it here.
\begin{algorithm*}[!ht]\caption{Matrix Sensing with Stochastic Gradient Descent (General Measurements).}\label{alg:stochastic_gradient_descent_general}
\begin{algorithmic}[1]
\Procedure{SGD\_General}{$\{u_i,b_i\}_{i\in [m]}$} \Comment{Lemma~\ref{lem:gd_cost_per_iter_general}}
    \State $\tau \gets \max_{i \in [m]} b_i $
    \State $A_1 \gets \tau \cdot I$
    \State $z_i\gets u_i^\top A_1 u_i - b_i$ for $i\in [m]$\Comment{$z\in \R^m$}
    \State $w_{i,j}\gets \langle u_i, u_j\rangle$ for $i,j\in [m]$ \Comment{$w\in \R^{m\times m}$}
    \For{$t = 1 \to T$}
        \State Sample ${\cal B}_t \subset [m]$ of size $B$ uniformly at random
        \State $\nabla \Phi_{\lambda} (A_t, {\cal B}_t) \gets \frac{m}{B} \sum_{i \in {\cal B}_t} u_i u_i^\top \lambda \sinh( \lambda z_{i} ) $\Comment{It takes $O(Bn^2)$-time}
        \State $\| \nabla \Phi_{\lambda} (A_t) \|_F\gets \lambda \left(\sum_{i,j=1}^m  w_{i,j}^2\sinh (\lambda z_{i})\sinh(\lambda z_j)\right)^{1/2}$\Comment{It takes $O(m^2)$-time}
        \State $A_{t+1} \gets A_t - \epsilon \cdot \nabla \Phi_{\lambda}(A_t,{\cal B} _t) / \| \nabla \Phi_{\lambda}(A_t) \|_F$\Comment{It takes $O(n^2)$-time}
        \For{$i\in [m]$}\Comment{Update $z$. It takes $O(mB)$-time}
            \For{$j\in {\cal B}_t$}
                \State $z_i\gets z_i - \epsilon  \lambda m w_{i,j}^2\sinh(\lambda z_{j}) / (\| \nabla \Phi_{\lambda}(A_t) \|_F B)$
            \EndFor
        \EndFor
    \EndFor
    \State \Return $A_{T+1}$
\EndProcedure
\end{algorithmic}
\end{algorithm*}





\end{document}